\documentclass[10pt,a4paper]{article}
\usepackage[latin1]{inputenc}

\usepackage{amsmath}
\usepackage{amsfonts}
\usepackage{amssymb}

\usepackage{amsthm}
\usepackage{graphicx}

\usepackage{geometry}

\usepackage[hidelinks]{hyperref}
\hypersetup{colorlinks=false}

\theoremstyle{plain}
\newtheorem{theorem}{Theorem}
\newtheorem{lemma}{Lemma}
\newtheorem{corollary}{Corollary}

\theoremstyle{definition}
\newtheorem{example}{Example}

\usepackage{mathtools}
\usepackage{enumitem}
\usepackage[ruled,vlined]{algorithm2e}

\usepackage{natbib}

\usepackage{subcaption}
\usepackage{tikz}
\usetikzlibrary{positioning, calc, arrows.meta, chains,fit,shapes.misc}

\usepackage{tabularx}
\usepackage{pgfplots, pgfplotstable}
\pgfplotsset{compat=1.5}
\usepackage{booktabs}
\usepackage{xstring}

\allowdisplaybreaks

\usepackage[title]{appendix}

\usepackage{authblk}

\author[a]{Alessandro Agnetis}
\author[b]{Ben Hermans\thanks{The author is funded by a Postdoctoral Fellowship of the Research Foundation -- Flanders.}\textsuperscript{,}}
\author[b]{Roel Leus}
\author[c]{Salim Rostami}

\affil[a]{Dipartimento di Ingegneria dell'Informazione e Scienze Matematiche, Universit\`a di Siena, Italy}
\affil[b]{Research Center for Operations Research \& Statistics, KU~Leuven, Belgium}
\affil[c]{I\'ESEG School of Management, France}
\date{} 
\title{Time-critical testing and search problems}

\begin{document}

\maketitle
\begin{abstract}
\noindent This paper introduces a problem in which the state of a system needs to be determined through costly tests of its components by a limited number of testing units and before a given deadline. We also consider a closely related search problem in which there are multiple searchers to find a target before a given deadline. These natural generalizations of the classical sequential testing problem and search problem are applicable in a wide range of time-critical operations such as machine maintenance, diagnosing a patient, and new product development. We show that both problems are NP-hard, develop a pseudo-polynomial dynamic program for the special case of two time slots, and describe a partial-order-based as well as an assignment-based mixed integer program for the general case. Based on extensive computational experiments, we find that the assignment-based formulation performs better than the partial-order-based formulation for the testing variant, but that this is the other way round for the search variant. Finally, we propose a pairwise-interchange-based local search procedure and show that, empirically, it performs very well in finding near-optimal solutions.
\end{abstract}

\section{Introduction}\label{sec:intro}
The timely diagnosis of a system constitutes a key task in a wide range of time-critical operations. The downtime of a machine due to a periodic or corrective maintenance, for example, crucially depends on the time needed to test which components are functioning correctly. Medical tests oftentimes also have a high time criticality, not only to establish a correct diagnosis, but also to identify an effective treatment. When developing a new product, in turn, it is essential to complete all tests concerning the product's safety and performance in time, because a delayed product launch could enable competitors to grasp a leading market share with a similar product.

Each of the above settings can be generically modeled as a multi-component system whose state is either up or down (healthy or unhealthy; safe or unsafe) depending on which of its~$n$ components are functioning \citep{unluyurt2004sequential}. Consider, for example, a machine that can only operate if all its components are functioning, a patient who only passes a medical examination if all test results are negative, or a new product that can only be launched if all safety and performance tests prove to be successful. These are all examples of so-called \emph{serial} or \emph{$n$-out-of-$n$ systems}, in which the state is only up if all of its components are functioning. A \emph{parallel} or \emph{$1$-out-of-$n$ system}, on the other hand, is up if there is at least one functioning component. For instance, it suffices to find a single treatment that works to cure a patient.

The \emph{sequential testing problem} deals with checking the state of a system through costly tests of its components, where each component has a given failure probability \citep{unluyurt2004sequential}. Failures are typically assumed to occur independently with a probability that is known from, for example, the analysis of historical data. The cost to test a component can measure both monetary expenses as well as intangible aspects such as how inconvenient the test is to a patient. Once the system's state is known, the testing procedure halts. The goal is then to determine in which sequence to test the different components so as to minimize the total expected testing cost.

In this paper, we consider a variant of the classical sequential testing problem in which every test takes unit time and in which the system's state has to be known within a given deadline~$T$. In order to comply with such timeliness requirements, there are~$m$ testing units on which tests can take place simultaneously, and the testing procedure is thus no longer purely sequential. As soon as the system's state is known, all testing units are halted.  This leads to what we call the \emph{time-critical testing problem}, where the goal is to schedule the~$n$ tests on the~$m$ available testing units so as to determine the system's state within the deadline~$T$ at minimum expected cost. Although we focus on serial systems, all our results are also directly applicable to parallel systems. Indeed, if we invert the interpretation of a test's outcome and of the definition of a system's up and down state, then the testing of a $1$-out-of-$n$ system becomes equivalent to the testing of a $n$-out-of-$n$ system \citep{unluyurt2004sequential}.

We also study a closely related \emph{time-critical search problem} in which there are~$m$ searchers and~$n$ possible locations that could contain a given target, and the objective is to find the target before a deadline~$T$ at minimum expected cost. This is a variant of the time-critical testing problem where the tests are dependent in the sense that exactly one of the system's components is defective, and where we need to identify this component \citep{wagner2001discrete}. Applications to this problem include a malfunctioning component of a machine that needs to be identified before a certain deadline, a terrorist that needs to be captured before an attack takes place, or a person in distress who needs to be found in the context of a search and rescue operation. 

The main contributions of this paper are as follows. We formally define the time-critical testing problem and the time-critical search problem, and establish that both problems are strongly NP-hard in general. We show in particular that the time-critical search problem reduces to the time-critical testing problem, which extends an earlier result of~\cite{kelly1982remark} towards multiple testing units. For the special case where $T = 2$, we establish weak NP-hardness and describe a pseudo-polynomial dynamic program as well as a fully polynomial-time approximation scheme (FPTAS).  Next, we introduce a partial-order-based and an assignment-based mixed integer programming (MIP) formulation for each of the problem variants, and show how to adapt these formulations to solve the precedence-constrained sequential testing problem \citep{monma1979sequencing} and the batch-testing problem \citep{daldal2016approximation}. To the best of our knowledge, no compact exact MIP formulations were previously available for the latter two problems. Finally, we illustrate that intuitive approaches such as a greedy heuristic or a pairwise-interchange-based local search procedure fail to provide a constant-factor approximation guarantee.

Based on extensive computational experiments, we find that the assignment-based formulation performs better than the partial-order-based formulation for the testing variant, while this is the other way round for the search variant.
Moreover, instances for the time-critical testing problem are considerably harder to solve than those for the search variant. Despite its lack of a theoretical worst-case guarantee, we also find that, empirically, our local search procedure seems to perform very well in finding near-optimal solutions in limited computation times.

The remainder of this paper is organized as follows. After the literature review in Section~\ref{sec:lit_rev}, we formally define the time-critical testing problem and search problem in Section \ref{sec:def}. Next, we examine the structural properties of an optimal schedule in Section~\ref{sec:struct_prop}, analyze the problems' complexity in Section~\ref{sec:compl}, and consider the special case of two time slots in Section~\ref{sec:DP}. Sections~\ref{sec:form} and~\ref{sec:heur}, in turn, discuss our proposed MIP formulations and local search heuristic, whose performance is evaluated in the computational experiments reported in Section~\ref{sec:comput}. Section~\ref{sec:conc}, finally, concludes and indicates directions for further research.

\section{Related work}\label{sec:lit_rev}
Since the U.S.\ Air Force highlighted the importance of the sequential testing problem \citep{johnson1956optimal}, it has received considerable attention in the literature. For a serial system without precedence constraints, the following intuitive result has been established independently in a variety of different contexts \citep{mitten1960analytic,boothroyd1960least,kadane1969quiz}: every sequence that inspects the components in non-decreasing order of their cost-to-failure-probability ratio is optimal. This idea has subsequently been generalized into polynomial-time algorithms for the sequential testing problem with series-parallel precedence constraints \citep{garey1973optimal,monma1979sequencing,berend2014optimal}. For general precedence constraints, the problem is NP-hard \citep{kelly1982remark,de2008r};  \cite{rostami2019sequential} describe a branch-and-price algorithm  as well as a dynamic program for this case. We refer to \cite{unluyurt2004sequential} for an excellent and detailed review of the sequential testing literature.

Most closely related to our work is the \emph{batch-testing problem} for serial systems as introduced by \cite{daldal2016approximation}. In this variant of the sequential testing problem, tests can be performed simultaneously in batches and there is a fixed set-up cost per batch. \cite{daldal2016approximation} develop an approximation algorithm for the problem with a~$6.829 + \varepsilon$ worst-case guarantee factor and an integer program that approximates the optimal solution within a factor~$1 + \varepsilon$, where~$\varepsilon \in (0,1)$ is arbitrary but fixed. \cite{daldal2017sequential} study the same problem with additional restrictions on the set of tests that can be combined in a batch, and \cite{segev2018polynomial} use techniques from approximate dynamic programming to devise a polynomial-time approximation scheme (PTAS) for the batch-testing problem. Contrary to the aforementioned articles, we consider a deadline, an upper bound~$m$ on the batch size, and no fixed set-up cost per batch. For this reason, the approximation algorithms as proposed by \cite{daldal2016approximation} and \cite{segev2018polynomial} do not seem to be directly modifiable towards our setting. The assignment-based formulation to be described in Section~\ref{sec:form}, however, can be adapted such that it also solves the batch-testing problem.

\cite{agnetis2009sequencing} define a problem where~$n$ unreliable jobs with given revenues are to be processed on~$m$ machines. If a job fails, then all subsequent jobs on the same machine are blocked, and their revenue is lost. The goal is then to assign the jobs to machines so as to maximize the expected revenue. \cite{agnetis2009sequencing} give a polyhedral characterization for the case of one machine and show that the problem is NP-hard for two machines. In subsequent articles, \cite{agnetis2014list} show that a list scheduling algorithm based on the so-called Z-ratio leads to a $0.8535$-approximation for the case of two machines, and to a $0.8531$-approximation in the general case of~$m$ machines \citep{agnetis2020largest}. Our work differs from the aforementioned articles since we minimize the expected cost rather than maximize the expected revenue, and since we assume that the testing procedure halts on all units as soon as the system's state is known.

Another related problem is the \emph{salvo policy problem}, in which there are a limited number of moments in time (the so-called salvos) at which multiple missiles can be fired simultaneously to stop an incoming object \citep{gould1984efficient}.  If all the missiles fired at a certain salvo are unsuccessful, the missiles scheduled at the next salvo must be fired. Each missile fired at a certain salvo has a hit probability that only depends on the salvo, and the goal is to minimize the expected number of missiles that need to be fired. \Citet{van2020efficient} develops algorithms to solve different variants of the problem, and we refer to his article as well as the work of \cite{glazebrook2004shoot} for a more detailed literature review. The main difference with our work is that, in our setting, the failure probabilities and testing costs are time-independent and can differ between components. Hence, we can also interpret the time-critical testing problem as a variant of the salvo policy problem where there are~$T$ salvos and $n$ missiles with different time-independent costs and hit probabilities, and where no more than $m$ missiles can be fired at each salvo. To the best of our knowledge, this problem variant has not been studied before.

If we identify success probabilities with job weights and search costs with job processing times, then the classical sequential search problem (i.e., the time-critical search problem with a single searcher and no deadline) is equivalent to the problem $1 \mid \mid \sum w_j C_j$ of minimizing the total weighted completion time of a job set on a single machine. Interestingly, numerous results for $1 \mid \mid \sum w_j C_j$ parallel those for the sequential testing problem: without precedence constraints it is optimal to search the different locations in non-decreasing order of their cost-to-probability ratio \citep{smith1956various,gluss1959optimum}, this rule can be generalized to solve the case with series-parallel precedence constraints efficiently \citep{lawler1978sequencing,sidney1975decomposition}, and the problem is strongly NP-hard for general precedence constraints \citep{lenstra1978complexity}. In fact, \cite{kelly1982remark} shows that the sequential testing problem is at least as hard as the problem $1 \mid \mid \sum w_j C_j$. Our results generalize these findings to a time-critical setting with multiple testing units, since we show in Section~\ref{sec:compl} that the time-critical search problem reduces to the time-critical testing problem.

In the rich literature on search theory (see for instance the books of \cite{stone2007theory} or \cite{alpern2013search} for a detailed overview), there are a number of articles that consider search problems with multiple searchers \citep{li2018multiple}, general cost and weight functions \citep{fokkink2019submodular,happach2020general}, or an unknown deadline \citep{lin2016finding,lidbetter2020search}. These problems, however, are overall quite different from our setting, and the time-critical search problem as defined in this paper appears to be novel.

\section{Problem definition and notation}\label{sec:def}
Consider an~$n$-out-of-$n$ system whose state is up if and only if all of its~$n$ components are functioning; otherwise, the system is down. Each component $j \in N \coloneqq \{1,\ldots, n\}$ is functioning with a given \emph{success probability} $p_j \in [0,1] \cap \mathbb{Q}$, where this event is independent of the outcome for other components, and we can determine whether a component is functioning by means of a test with \emph{testing cost}~$c_j \in \mathbb{N}$. We assume that a test requires one time unit, that the cost of a test is borne regardless of its outcome, and that tests are perfect in the sense that neither type-1 nor type-2 errors may occur.  There are $m \leq n$ identical \emph{testing units} or \emph{machines} on which tests can take place simultaneously, and we require that all tests be completed within a \emph{deadline}~$T \leq n$. To ensure feasibility, we assume that~$n \leq mT$. If a test fails, then the process stops on all machines and we conclude that the system is down. The objective is to schedule all~$n$ tests on the~$m$ available machines so as to determine the system's state at minimum expected cost within the deadline~$T$. 

Since all tests require unit time, we can divide the time horizon into~$T$ equal \emph{time slots}. Moreover, since the testing procedure is halted on all machines as soon as the system's state is known, it is not relevant which machine a given test is assigned to. Hence, we can represent a feasible solution or \emph{schedule} by an ordered partition $\sigma = (S_1,\dots,S_T)$ of the set $N$ into $T$ sets, where~$S_t \subseteq N$ with $|S_t|\leq m$ denotes the set of tests scheduled in time slot $t = 1,\dots,T$. 
For an arbitrary schedule~$\sigma = (S_1,\dots,S_T)$, denote by $\mathcal{A}_t \coloneqq \bigcup_{k=1}^{t-1} S_k$ the set of all tests scheduled before time slot $t \in \{1, \ldots, T\}$, where $\mathcal{A}_1 = \emptyset$. Given a subset~$S \subseteq N$ and a sequence $(a_j)_{j \in N}$ of real numbers, we use the conventional notation that~$a(S) \coloneqq \sum_{j \in S} a_j$, and we let~\mbox{$\sum_{j \in \emptyset} a_j = 0$} and \mbox{$\prod_{j \in \emptyset} a_j = 1$}. 
The \emph{time-critical testing problem} (TCTP) then asks for a schedule~$\sigma = (S_1, \ldots, S_T)$ that minimizes the \emph{expected testing cost}
\begin{equation}
z(\sigma) \coloneqq \sum_{t=1}^T c(S_t) \prod_{j\in \mathcal{A}_{t}}p_j. \label{eq:expcost}
\end{equation}
Here, each index~$t$ associates the cost $c(S_t)$ with the probability that the tests in~$S_t$ need to be performed because all components tested before time~$t$ were successful (and the system's state is thus still unknown). A TCTP instance is completely defined by the tuple~$(m, n, T, (c_j, p_j)_{j \in N})$, and we refer to~$m$-TCTP as the special case of the problem when the number of machines is fixed to~$m$.

We also consider a search problem related to TCTP, where a given target is hidden in \emph{exactly one} out of~$n$ possible locations. Each location $j \in N \coloneqq \{1,\ldots,n\}$ has a probability~$\pi_j \in [0,1] \cap \mathbb{Q}$ to contain the target, where $\sum_{j \in N} \pi_j = 1$. Searching location~$j$ takes unit time and leads to a cost~$c_j \in \mathbb{N}$. Given that there are~$m \leq n$ available searchers and that the object needs to be found before a certain deadline~$T\leq n$, where~$n \leq Tm$, the question is when to search each location so as to minimize the expected cost.
Using the same notation as above, a feasible schedule can be represented by an ordered partition~$\sigma = (S_1,\ldots, S_T)$ of the set~$N$ into $T$ sets with~$\vert S_t \vert \leq m$ for each~$t = 1,\ldots,T$. The \emph{time-critical search problem} (TCSP) then asks for a schedule $\sigma=(S_1,\dots,S_T)$ that minimizes the \emph{expected search cost} 
\begin{equation}\label{eq:search_cost}
\bar{z}(\sigma)	\coloneqq \sum_{t = 1}^T c(S_t) \sum_{k = t}^T \pi(S_k).
\end{equation}
Here, the index~$t$ in the outer summation reflects the fact that we incur cost~$c(S_t)$ when the target is in one of the locations searched in time slots $k = t, \ldots, T$, which occurs with probability $\sum_{k = t}^T \pi(S_k)$. A TCSP instance is completely defined by the tuple~$(m, n, T, (c_j, \pi_j)_{j \in N})$, and we refer to \mbox{$m$-TCSP} as the special case with a fixed number of~$m$ searchers.

The difference between TCTP and TCSP consists in the way in which both variants model the probability that a job (i.e., testing a component or searching a location) needs to performed. While for TCTP the success probabilities are taken to be independent, we assume for TCSP that exactly one location contains the target and the probabilities thus add up to one. For TCTP, the probability that a job needs to be performed then equals the product of success probabilities over all jobs scheduled in preceding time slots. For TCSP, in turn, this probability equals one minus the probability that one of the locations searched in preceding time slots contains the target.

\section{Structural properties}\label{sec:struct_prop}
If there is only a single testing unit or searcher, then TCTP and TCSP reduce to the sequential testing problem and to the problem $1 \mid \mid \sum w_j C_j$ of minimizing the weighted completion time on a single machine, respectively. Hence, as mentioned in Section~\ref{sec:lit_rev}, every sequence that tests the components or searches the locations in non-decreasing order of, respectively, the ratio~$c_j / (1-p_j)$ or~$c_j / \pi_j$ is optimal. Lemma~\ref{lem:non_empt_time}, stated below, implies that the resulting schedule is also optimal for the setting with multiple machines or searchers if the deadline allows for this. That is, the only reason to perform tests or to search locations simultaneously is to comply with a deadline~$T < n$.

\begin{lemma}\label{lem:non_empt_time}
	For each TCTP or TCSP instance there exists an optimal schedule $\sigma^\star = (S^\star_1, \ldots, S^\star_T)$ in which~$S^\star_t$ is non-empty for all $t \leq T$.
\end{lemma}
\begin{proof}
	Since the proof is completely analogous for TCTP and TCSP, we only consider the former case. For an arbitrary TCTP instance, let $\sigma = (S_1, \ldots, S_T)$ be an optimal schedule that does not satisfy the property. By a straightforward pairwise interchange argument on Equation~\eqref{eq:expcost}, we can assume without loss of generality that all non-empty sets are scheduled before the empty sets in~$\sigma$. Call~$t$ the first time slot for which $\vert S_{t} \vert \geq 2$; such a slot must exist because otherwise~$\sigma$ would satisfy the property (recall that our definition of a TCTP instance assumes that~$T \leq n$). Now take an arbitrary test~$j \in S_{t}$ and define a new schedule~$\sigma^\prime = (S_1, \ldots, S_{t} \setminus \{j\}, \{j\}, S_{t + 1}, \ldots, S_{T-1})$. Equation~\eqref{eq:expcost} then yields that
	\begin{equation*}
		z(\sigma) - z(\sigma^\prime) = \left(\prod_{i \in \mathcal{A}_{t}} p_i\right) \left(1 - \prod_{i \in S_{t} \setminus \{j\}} p_i \right) c_j \geq 0.
	\end{equation*}
	Repeating this argument at most~$n$ times yields an optimal schedule that satisfies the property.
\end{proof}

The following lemma generalizes the ratio rule to the setting of multiple machines. In particular, given a subset~$S \subseteq N$ for a TCTP instance, define the \emph{cost-to-probability ratio} as
\begin{equation*}
\rho(S) \coloneqq \frac{c(S)}{1 - \prod_{j \in S}p_j}
\end{equation*}
if $\prod_{j \in S}p_j \neq 1$, and~$\rho(S) \coloneqq +\infty$ otherwise. For an instance of TCSP, in turn, we define the ratio
\begin{equation*}
\bar{\rho}(S) \coloneqq \frac{c(S)}{\pi(S)}
\end{equation*}
if $\pi(S) \neq 0$, and~$\bar{\rho}(S) \coloneqq +\infty$ otherwise. Lemma~\ref{lem:struct_prop} below implies that, for both TCTP and TCSP, there exists an optimal schedule that sequences the different sets in non-decreasing order of the corresponding ratios.

\begin{lemma}\label{lem:struct_prop}
	Reordering the sets~$S_t$ of a schedule~$\sigma = (S_1,\ldots, S_T)$ in non-decreasing order of the ratio~$\rho(S_t)$ or~$\bar{\rho}(S_t)$ does not increase the expected testing or search cost, respectively.
\end{lemma}
\begin{proof}
	Since the proof is again completely analogous for TCTP and TCSP, we only consider the former case. Let $\sigma = (S_1, \ldots, S_T)$ be a schedule with $\rho(S_t) > \rho(S_{t+1})$ for some~$t < T$ and define the schedule~$\sigma^\prime$ from~$\sigma$ in which~$S_t$ and~$S_{t+1}$ have been interchanged. Equation~\eqref{eq:expcost} then yields that
	\begin{equation*}
	z(\sigma) - z(\sigma^\prime) = \left(\prod_{j \in \mathcal{A}_{t}} p_j\right) \left[ \left(1 - \prod_{j \in S_{t + 1}} p_j \right) c(S_{t}) - \left(1 - \prod_{j \in S_{t}} p_j \right) c(S_{t + 1}) \right],
	\end{equation*}
	which is non-negative. Repeating this procedure at most~$T^2$ times then yields the result.
\end{proof}

Based on Lemma~\ref{lem:struct_prop}, one might propose a greedy approach where we test in each time slot a subset~$S$ of the remaining components for which the ratio~$\rho(S)$ is minimal among all sets that contain sufficient tests to meet the deadline. The next example, however, shows that this approach does not lead to a constant-factor approximation guarantee for TCTP.\@  A similar example shows that an analogous greedy approach that minimizes the ratio~$\bar{\rho}(S)$ also fails to provide a constant-factor approximation guarantee for TCSP.\@

\begin{example}\label{ex:ratio}
	Consider a TCTP instance  with~$n = 3$ tests, $m = 2$ machines, and deadline~$T=2$. Given an integer~$M>0$, we define the costs and success probability of each job as follows: $c_1 = 1$, $p_1 = 1/M$,  $c_2= 0$,  $p_2 = 1-1/M$,  $c_3 = M$,  and $p_3  = 1-1/M$. The greedy approach described above would then test component~2 in the first time slot (its ratio equal to~$0$ is clearly minimal), and the remaining components~$\{1,3\}$ in the second time slot. This leads to an expected testing cost equal to $z_\text{greedy} = c_2 + p_2 (c_1 + c_3) = 0 + (1-1/M) (M + 1) = M - 1/M$. The schedule that tests~$\{1\}$ first and~$\{2,3\}$ next, however, attains an expected testing cost equal to $z^\star = c_1 + p_1 (c_2 + c_3) = 1 + (1/M) M = 2$. Hence, the ratio $z_\text{greedy} / z^\star$ is $\Omega(M)$ and thus cannot be bounded by a constant.
\end{example}

The next lemma reveals that, for TCTP, the range between the minimal and maximal expected testing cost, and thus the potential benefit of identifying an optimal schedule, increases as the joint success probability~$\prod_{j \in N}p_j$ decreases. Intuitively, if the system has a higher probability to be up, then it becomes more likely that all tests need to be performed, and their specific ordering becomes less influential. Observe that an analogous result does not hold for TCSP, since the probability~$\sum_{j \in N} \pi_j$ is fixed to one  for this problem.

\begin{lemma}\label{lem:trivial_approx}
	Given a TCTP instance, let~$\sigma^\star = (S^\star_1, \ldots, S^\star_T)$ be an optimal schedule. For every schedule~$\sigma=(S_1, \ldots, S_T)$ it then holds that~$z(\sigma^\star) \geq z(\sigma) \prod_{j \in N}p_j $.
\end{lemma}
\begin{proof}
	Consider an arbitrary schedule $\sigma=(S_1, \ldots, S_T)$. Since~$z(\sigma) \leq c(N)$, we obtain that
	\[
	z(\sigma^\star) = \sum_{t=1}^T c(S^\star_t) \prod_{j \in \mathcal{A}^\star_{t}} p_j \geq \sum_{t=1}^T c(S^\star_t) \prod_{j \in N} p_j = c(N) \prod_{j \in N} p_j \geq z(\sigma) \prod_{j \in N} p_j. \qedhere
	\]
\end{proof}

\section{Complexity analysis}\label{sec:compl}
In this section we analyze the complexity of TCTP and TCSP.\@ First, we show in Section~\ref{sec:search_prob} that TCSP is NP-hard even if~$T = 2$ and that the \mbox{$m$-TCSP} is strongly NP-hard for every~$m \geq 3$. Next, we show that TCSP reduces to TCTP in Section~\ref{sec:TCTP_NPhard}, so that the hardness results for TCSP carry over directly to TCTP.\@

\subsection{Hardness of TCSP}\label{sec:search_prob}
Given a constant~$\alpha \in \mathbb{Q}$, consider the decision version of TCSP, which we call TCSP-D, that asks whether there exists a schedule~$\sigma$ such that~$\bar{z}(\sigma) \leq \alpha$. To establish the complexity of TCSP-D, we will use the following lemma:
\begin{lemma}\label{lem:concavity}
	For every $r \in \mathbb{N}$ numbers $a_1, \ldots, a_r \in \mathbb{R}$ with a fixed sum $A \coloneqq \sum_{i = 1}^r a_i$ it holds that
	\begin{equation*}
	\sum_{i = 1}^r \sum_{j = i}^r a_i a_j \geq \frac{(r+1)A^2}{2r},
	\end{equation*}
	with equality holding if and only if $a_i = A / r$ for every $i = 1, \ldots, r$.
\end{lemma}
\begin{proof}
	For arbitrary numbers~$r \in \mathbb{N}$ and~$A \in \mathbb{R}$, define the set~$X \coloneqq \{(x_1, \ldots, x_r) \in \mathbb{R}^r\colon \sum_{i=1}^r x_i = A\}$ and the function $f\colon X \to \mathbb{R}\colon (x_1, \ldots, x_r) \mapsto \sum_{i = 1}^r \sum_{j = i}^r x_i x_j$. It then suffices to show that~$f$ attains its unique minimum if $x_i = A/r$ for each $i = 1, \ldots, r$. To see that this is the case, observe that the function $g\colon \mathbb{R}^{r-1} \to \mathbb{R}\colon (x_2, \ldots, x_r) \mapsto f(A - \sum_{i = 2}^r x_i, x_2, \ldots, x_r)$ is strictly convex with partial derivatives equal to zero if and only if $x_j = A - \sum_{i = 2}^r x_i$ for each~$j = 2, \ldots, r$.
\end{proof}

TCSP-D is in NP since, for a given schedule~$\sigma$, we can use Equality~\eqref{eq:search_cost} to verify in polynomial time whether~$\bar{z}(\sigma) \leq \alpha$. The next theorem uses a reduction from the NP-complete \textsc{Partition} problem \citep{garey2002computers} to show that TCSP-D is NP-complete even for the special case of two time slots.

\begin{theorem}\label{th:DecProbNPcomp}
	TCSP-D is NP-complete, even if $T = 2$.
\end{theorem}
\begin{proof}
	Consider an instance of the \textsc{Partition} problem where, given~$q \in \mathbb{N}$ natural numbers $u_1,\ldots, u_q \in \mathbb{N}$ with $A \coloneqq \sum_{i=1}^q  u_i$, the question is whether there exists a set~$S \subseteq \{1,\ldots,r\}$ such that $u(S) = A/2$. Now define a TCSP-D instance with $m = q$ searchers, deadline~$T = 2$, and~$n = q$ locations with~$c_j = u_j$ and $\pi_j = u_j / A$ for all $j = 1, \ldots, q$. Finally, we let~$\alpha = 3A/4$. 
	
	For every schedule~$\sigma = (S_1, S_2)$, Lemma~\ref{lem:concavity} (with $r = 2$ and identifying $a_i$ with $u(S_i)$ for $i = 1,2$) yields that
	\begin{equation*}
	\sum_{t = 1}^2 c(S_t) \sum_{k = t}^2 \pi(S_k) = \frac{1}{A} \sum_{t = 1}^2 \sum_{k = t}^2 u(S_{t}) u(S_k) \geq \frac{1}{A} \frac{3 A^2}{4}=\frac{3A}{4}=\alpha.
	\end{equation*}
	Hence, the answer to the TCSP-D instance is `yes' if and only if there exists a schedule~$\sigma = (S_1,S_2)$ for which the above inequality holds with equality. From Lemma~\ref{lem:concavity}, this occurs if and only if $u(S_1)=u(S_2)=A/2$, i.e., if and only if the answer to the \textsc{Partition} instance is `yes'.
\end{proof}

The following theorem shows that TCSP-D, where~$T$ is part of the input, is strongly NP-complete for every fixed number~$m \geq 3$ of searchers. The proof uses a reduction from the strongly NP-complete \mbox{\textsc{3-Partition}} problem \citep{garey2002computers}.

\begin{theorem}\label{th:DecProbstrongNPcomp}
	TCSP-D with a fixed number~$m \geq 3$ of searchers is strongly NP-complete. 
\end{theorem}
\begin{proof}
	The proof uses a reduction from the \textsc{3-Partition} problem. Given $q,B \in \mathbb{N}$, consider~$3q$ integers $u_1, \ldots, u_{3q}$ with~$\sum_{j=1}^{3q}u_j = qB$ and~$B/4 < u_j < B/2$ for each~$j = 1, \ldots, 3q$. The \mbox{\textsc{3-Partition}} problem then asks whether $\{u_1, \ldots, u_{3q}\}$ can be partitioned into~$q$ subsets $U_1, \ldots, U_q$ with equal sum $u(U_t) = B$ for each $t = 1, \ldots, q$. Observe that, since~$B/4 < u_i < B/2$ for each~$i = 1, \ldots, 3q$, the subset~$U_t$ must have equal cardinality~$\vert U_t \vert = 3$ for each $t=1,\ldots, q$.
	
	For an arbitrary \textsc{3-Partition} instance $(q, B, u_1, \ldots, u_{3q})$, define a \mbox{TCSP-D} instance with $\alpha = (q + 1)B/2$, a fixed number of $m \geq 3$ searchers, a deadline~$T = q$, and $n = 3q$ locations with $c_j = u_j$ and~$\pi_j = u_j / (qB)$ for every $j \in N$.
	
	For every schedule~$\sigma = (S_1, \ldots, S_T)$ it follows from Lemma~\ref{lem:concavity} (by taking $r = T = q$, and by identifying the numbers $a_i$ with $u(S_i)$ for $i = 1,\ldots, T$) that
	\begin{equation*}
	\sum_{t = 1}^T c(S_t) \sum_{k = t}^T \pi(S_k) =  \frac{1}{qB} \sum_{t = 1}^T  \sum_{k = t}^{T} u(S_t) u(S_{k}) \geq \frac{1}{qB} \frac{(q+1)(qB)^2}{2q} = \frac{(q+1)B}{2} = \alpha.
	\end{equation*}
	Hence, the answer to the TCSP-D instance is  `yes' if and only if there exists a schedule~$\sigma = (S_1, \ldots, S_T)$ for which the above inequality holds at equality. From Lemma~\ref{lem:concavity}, this occurs if and only if $u(S_t)=qB/q=B$ for all $t=1,\dots,T$. Observe that, since $B/4 < u_i < B/2$ for each $i = 1, \ldots, 3q$, this latter condition can only be satisfied if~$\sigma$ schedules exactly three jobs in every time slot, such that the schedule is feasible with respect to the machine capacity~$m \geq 3$.
\end{proof}

\subsection{Hardness of TCTP}\label{sec:TCTP_NPhard}
We define the decision version TCTP-D of TCTP as follows: given a constant~$\beta \in \mathbb{Q}$, does there exist a schedule~$\sigma$ such that $z(\sigma) < \beta$? To show that TCSP-D reduces to TCTP-D in Theorem~\ref{th:DecProb_to_TCTP}, we apply a similar approach as the one used in \cite{hermans2019timely} to decompose and to bound a product of probabilities. In particular, we need the following lemma, which can be shown by induction.

\begin{lemma}[\citeauthor{hermans2019timely},~\citeyear{hermans2019timely}]\label{lem:decomposition}
	Let $r \in \mathbb{N}$ and $a_1, \ldots, a_r \in \mathbb{R}$, then
	\begin{equation*}
	\prod_{i = 1}^{r} (1-a_i) = 1 - \sum_{i = 1}^{r} a_i + \sum_{i = 1}^{r} \sum_{j = i + 1}^{r} a_i a_j \prod_{k = j + 1}^{r} (1-a_k).
	\end{equation*} 
\end{lemma}

\begin{theorem}\label{th:DecProb_to_TCTP}
	TCSP-D reduces to TCTP-D.
\end{theorem}
\begin{proof}
	Given an instance of TCSP-D  $(\alpha, m, n, T, (c_j, \pi_j)_{j \in N})$, let~$W \in \mathbb{N}$ be the smallest integer such that~$w_j \coloneqq \pi_j W$ is integer for all~$j \in N$. Since all probabilities~$(\pi_j)_{j \in N}$ are rational numbers,~$W$ is polynomially bounded in the input size. Define an associated TCTP-D instance with $p_j = 1 - w_j / M$ for every $j \in N$, where $M \coloneqq c(N) (n W)^2$, and with all other parameters equal to the corresponding ones of the TCSP-D instance. Moreover, let $\gamma \coloneqq \left(c(N) - \alpha \right) W$ and $\beta = c(N) - (\gamma - 1) / M$.
	We want to show that there exists a schedule~$\sigma$ with $\bar{z}(\sigma) \leq \alpha$ if and only if there exists a schedule~$\sigma$ for the TCTP instance such that $z(\sigma) < \beta$.
	
	Let~$\sigma = (S_1, \ldots, S_T)$ be an arbitrary schedule, then it follows from Equation~\eqref{eq:expcost} and the definition of~$(p_j)_{j \in N}$ that
	\begin{equation}\label{eq:z_w}
		z(\sigma) = \sum_{t = 1}^T c(S_t) \prod_{j \in \mathcal{A}_{t}} \left(1 - \frac{w_j}{M}\right).
	\end{equation}
	Now define for every set~$S \subseteq N$ the number
	\begin{equation}
		Q(S) = \sum_{i \in S} \sum_{j \in S\colon j > i} \frac{w_i w_j}{M^2} \prod_{k \in S \colon k > j} \left(1-\frac{w_k}{M}\right), \label{eq:Q_S}
	\end{equation}
	and note that, by choice of~$M$, it holds for each $i,j,k \in S$ that
	\begin{equation*}
	0 \leq \frac{w_i w_j}{M} <  \frac{1}{c(N)n^2} \quad \text{ and } \quad 0 < 1 - \frac{w_k}{M} \leq 1.
	\end{equation*}
	Since Equation~\eqref{eq:Q_S} contains strictly less than~$n^2$ terms, this implies that $0 \leq Q(S) < 1/(Mc(N))$ for each $S \subseteq N$. Hence, by using Lemma~\ref{lem:decomposition} to rewrite Equation~\eqref{eq:z_w}, we obtain that
	\begin{align}
		z(\sigma) &= \sum_{t = 1}^T c(S_t) \left(1 - \frac{1}{M}w(\mathcal{A}_{t}) + Q(\mathcal{A}_{t})\right) \nonumber\\
		&< c(N) + \frac{1}{M} - \frac{1}{M} \sum_{t = 2}^T c(S_t) w(\mathcal{A}_{t}) \label{eq:z_UB}
	\end{align}
	and
	\begin{equation}
		z(\sigma) \geq c(N) - \frac{1}{M} \sum_{t = 2}^T c(S_t) w(\mathcal{A}_{t}). \label{eq:z_LB}
	\end{equation}
	
	If $\sum_{t = 2}^T c(S_t) w(\mathcal{A}_{t}) \geq \gamma$, then Inequality~\eqref{eq:z_UB} implies that $z(\sigma) < c(N) - (\gamma - 1) / M = \beta$. Conversely, if $\sum_{t = 2}^T c(S_t) w(\mathcal{A}_{t}) < \gamma$, then we also know that $\sum_{t = 2}^T c(S_t) w(\mathcal{A}_{t}) \leq \gamma - 1$ since $(c_j)_{j \in N}$ and $(w_j)_{j \in N}$ are integer. Inequality~\eqref{eq:z_LB} then yields that~$z(\sigma) \geq c(N) - (\gamma - 1) / M = \beta$. Hence, we obtain that 
	\begin{equation}\label{eq:compl_gamma}
	\sum_{t = 2}^T c(S_t) \sum_{k = 1}^{t-1} w(S_k) = \sum_{t = 2}^T c(S_t) w(\mathcal{A}_{t}) \geq \gamma = (c(N) - \alpha) W 
	\end{equation}
	if and only if $z(\sigma) < c(N) - (\gamma - 1) / M = \beta$.
	
	Finally, Equation~\eqref{eq:search_cost} combined with~$\pi(N) = \sum_{k=1}^T \pi(S_k) = 1$ and~$w_j = \pi_j W$ yields that 
	\begin{equation*}
		\left(c(N) - \bar{z}(\sigma)\right)W = \sum_{t=1}^T c(S_t) \left(\sum_{k=1}^T \pi(S_k) - \sum_{k=t}^T \pi(S_k) \right)W = \sum_{t = 2}^T c(S_t) \sum_{k = 1}^{t-1} w(S_k).
	\end{equation*}
	A schedule~$\sigma$ thus satisfies Inequality~\eqref{eq:compl_gamma} if and only if $\bar{z}(\sigma) \leq \alpha$, which completes the proof.
\end{proof}

TCTP-D is in NP since, for a given schedule~$\sigma$, we can use Equality~\eqref{eq:expcost} to verify in polynomial time whether~$z(\sigma) < \beta$. Hence, together with Theorems~\ref{th:DecProbNPcomp},~\ref{th:DecProbstrongNPcomp}, and~\ref{th:DecProb_to_TCTP}, we obtain this section's main result.

\begin{corollary}\label{col:TCTP_NPhard}
	TCTP-D is NP-complete even if $T = 2$. If~$T$ is part of the input, then the problem is strongly NP-complete for every fixed number~$m \geq 3$ of machines.
\end{corollary}

\section{Special case of two time slots}\label{sec:DP}
We now develop a pseudo-polynomial dynamic program for the special case of TCTP where~$T = 2$ and show how it can be converted into a fully polynomial-time approximation scheme (FPTAS). Together with Corollary~\ref{col:TCTP_NPhard} this yields that TCTP is weakly NP-hard when $T = 2$. Our approach is inspired by the standard pseudo-polynomial dynamic program for the knapsack problem \citep{kellerer2004knapsack}. An analogous procedure also solves TCSP with~$T=2$, but we omit its description for the sake of brevity.

If $T  = 2$, then every feasible schedule is uniquely defined by the set that is tested in the first time slot. For an optimal sequence~$\sigma^\star = (S^\star, N \setminus S^\star)$, we refer to the set~$S^\star$ as an \emph{optimal testing set}. We assume that $n = mT$, which is without loss of generality because one can always add dummy tests with cost~$0$ and success probability~$1$ until the instance satisfies this condition. The goal of our dynamic program is to identify such an optimal testing set through a backward recursion based on the index number~$i \in \{1,\ldots, n\}$ of a job.

For every \emph{stage}~$i \in \{1, \ldots, n\}$ and \emph{state} $(b, s) \in \{0,\ldots, c(N)\}  \times \{0,\ldots, m\}$, we define the \emph{value function} $v_i(b,s)$ as the minimum joint probability~$\prod_{j \in S} p_j$ of a set~$S \subseteq \{i, \ldots, n\}$ with cost~$c(S) = b$ and cardinality~$\vert S \vert = s$. That is, among the jobs corresponding to the current stage and the higher-indexed stages, we want to select a subset such that, if we schedule these jobs in the first time slot, then the probability that we need to perform the tests in the second time slot is minimal. Here, the state~$(b,s)$ specifies the cost and the number of elements that this subset should attain. If~$b$ and~$s$ are such that no set $S \subseteq \{i, \ldots, n\}$ with~$c(S) = b$ and~$\vert S \vert = s$ exists, then we define~$v_i(b,s) \coloneqq +\infty$.

We now discuss how to compute the value function~$v_{i}(b,s)$ for every stage and state recursively. As the boundary condition, let~$v_{n + 1}(b,s) \coloneqq 1$ if $b = s =0$ and $v_{n + 1}(b,s) \coloneqq +\infty$ otherwise. For every stage~$i \in \{1, \ldots, n\}$ and state $(b, s) \in \{0,\ldots, c(N)\}  \times \{0,\ldots, m\}$, it then holds that
\begin{equation}
	v_i(b,s) = \min\{p_i v_{i + 1}(b - c_i, s - 1),\, v_{i + 1}(b,s)\} \label{eq:recursion}
\end{equation}
if $v_{i + 1}(b - c_i, s - 1)$ is defined and finite, and $v_i(b,s) = v_{i + 1}(b,s)$ otherwise.
Here, the option~$p_i v_{i + 1}(b - c_i, s - 1)$ corresponds to selecting test~$i$ in that stage and state, whereas the other option refers to not selecting test~$i$. Starting from stage~$n$, Equation~\eqref{eq:recursion} allows us to compute the value function~$v_1(b,s)$ for every state~$(b,s)$ through a backward recursion. By keeping track of the selected tests, we also obtain a testing set that minimizes the joint probability.

Once we know the value function~$v_1(b,m)$ for every~$b \in \{0,1,\ldots, c(N)\}$, we can solve TCTP with $T = 2$ by enumerating over all these possible costs~$b$ and selecting one~$b^\star$ that minimizes~$b + v_{1}(b,m) (c(N) - b)$. The accompanying set of tests that attains this cost~$b^\star$ with minimum joint probability~$v_{1}(b^\star,m)$ then forms an optimal testing set. Observe that we can indeed limit ourselves to states~$(b,s)$ with $s = m$ in the first stage since we know that exactly~$m$ jobs need to be scheduled in the first time slot as a consequence of our assumption that~$n = mT = 2m$. 

Since the computational effort depends on the number~$c(N)$, the dynamic program has a pseudo-polynomial running time. By appropriately rounding the costs~$(c_j)_{j \in N}$, however, we can convert this pseudo-polynomial dynamic program into a FPTAS. This procedure is analogous to the standard scaling approach for the knapsack problem \citep{kellerer2004knapsack}. Theorem~\ref{th:fptas}, proven in Appendix~\ref{app:fptas}, summarizes this result.

\begin{theorem}\label{th:fptas}
	Given a TCTP instance~$(m, n, T, (c_j, p_j)_{j \in N})$ with~$T = 2$, let~$\sigma^\star = (S^\star, N \setminus S^\star)$ be an optimal testing sequence. For each $\varepsilon > 0$, we can obtain a sequence $\sigma_\varepsilon = (S_\varepsilon, N \setminus S_\varepsilon)$ with $z(\sigma_\varepsilon) \leq (1 + \varepsilon) z(\sigma^\star)$ in time $O(n^5 / \varepsilon)$.
\end{theorem}
\begin{proof}
	See Appendix~\ref{app:fptas}.
\end{proof}

\section{Mixed integer programming formulations}\label{sec:form}
In this section, we describe a partial-order-based and an assignment-based mixed integer programming formulation for both TCTP and TCSP.\@ We first address the testing variant in detail and then indicate how the corresponding formulation can be adapted to also solve the search variant. The partial-order-based formulation can be easily modified to also incorporate precedence constraints and, as such, it can be used to solve the precedence-constrained sequential testing problem \citep{monma1979sequencing}. The assignment-based formulation, in turn, can be modified to incorporate a fixed cost for every used time slot and thus to solve the batch-testing problem \citep{daldal2016approximation}. This illustrates that our approach, based on a linearization of the objective function, is sufficiently general to be useful to solve related testing problems as well.

\subsection{Partial-order-based formulation}\label{sec:po_form}
Inspired by \citet{Potts1980} and \citet{wagner2001discrete}, our first formulation relies on binary decision variables that indicate the partial order defined by a schedule. Additionally, we include decision variables that keep track of the probability that a certain component (or location) needs to be tested (or searched). 

\subsubsection{Testing variant}
For every two tests~$i,j \in N$, define a binary decision variable~$\delta_{ij}$ that equals 1 if test~$i$ is performed before test~$j$ and 0 otherwise, and a binary decision variable~$\mu_{\{i,j\}}$ that equals 1 if tests~$i$ and~$j$ are performed simultaneously. The brackets in the subscript of~$\mu_{\{i,j\}}$ indicate that the order of indices~$i$ and~$j$ is of no importance. Next, for every test~$i \in N$ and integer~$j \in \{ 0, \ldots, n \}$, define a continuous decision variable~$\alpha_{ij}$ that equals the probability that test~$i$ needs to be performed, conditional on the information that all tests~$j + 1, \ldots, n$ will be successful. That is, given a set $S \subset N$ of tests performed before test~$i \in N$, we want for every $j=0,\ldots, n$ that~$\alpha_{ij} = \prod_{l \in S \colon l \leq j} p_l$ and, therefore, that $\alpha_{in}  = \prod_{l \in S} p_l$ equals the probability that test~$i$ needs to be performed.

The following mixed integer program then constitutes a valid formulation for TCTP if~$n = mT$.
\begin{subequations}
	\begin{align}
	& \text{min} & \sum_{i \in N} c_i \alpha_{in}  \label{mip:po_obj}\\
	&\text{s.t.} & \delta_{ij} + \delta_{ji} + \mu_{\{i,j\}} & =  1 & & \text{$\forall\, i,j \in N\colon i \neq j$} \label{mip:po_order}\\		
	& &  \mu_{\{i,j\}}+\delta_{ij} + \delta_{jk} - \delta_{ik} & \leq 1  & & \text{$\forall\, i,j,k \in N\colon i \neq j \neq k \neq i$} \label{mip:po_transd}\\
	& & \sum_{j \in N \setminus \{i\}} \mu_{\{i,j\}} &= m - 1 & & \text{$\forall\, i \in N$} \label{mip:po_machcap}\\
	& & \alpha_{i0} &= 1 & & \text{$\forall\, i\in N$} \label{mip:po_y0}\\
	& & \alpha_{ij} &\geq \alpha_{i,j-1} - \delta_{ji} & & \text{$\forall\, i\in N$ and $j \in \{1,\ldots, n\}$} \label{mip:po_y}\\
	& & \alpha_{ij} &\geq p_{j} \alpha_{i,j-1} & & \text{$\forall\, i\in N$ and $j \in \{1,\ldots, n\}$} \label{mip:po_yp}\\
	& &  \delta_{ij},\, \mu_{\{i,j\}}  &  \in \{0,1\} & & \text{$\forall\, i,j \in N$} \label{mip:po_bin}
	\end{align}	
\end{subequations}

Objective function~\eqref{mip:po_obj} relates the probability that a test needs to be performed with its cost. Constraints~\eqref{mip:po_order}-\eqref{mip:po_machcap} together with the binary constraints enforce that the variables~$\delta_{ij}$ and~$\mu_{\{i,j\}}$ define a partial order of the tests such that in each time slot exactly~$m$ tests are performed.  Constraints~\eqref{mip:po_y0}-\eqref{mip:po_yp} and the non-negative coefficients~$c_i$ in the objective function imply that~$\alpha_{ij}$ equals the probability that test~$i$ needs to be performed conditional on the information that all tests~$j + 1, \ldots, n$ will be successful, as desired. Indeed, if~$\alpha_{i,j-1}$ equals the probability that test~$i$ needs to be performed conditional on the information that all tests~$j, \ldots, n$ will be successful, then~$\alpha_{ij}$ equals~$p_{j} \alpha_{i,j-1}$ if test~$j$ is performed before test~$i$, and~$\alpha_{ij} = \alpha_{i,j-1}$ otherwise.

\subsubsection{Search variant}
For TCSP we adopt the same definition for~$\delta_{ij}$ and~$\mu_{\{i,j\}}$,  but now we directly let~$\alpha_{j}$ reflect the probability that location~$j \in N$ needs to be searched. That is,~$\alpha_{j}$ equals the probability that the target has not been found at the time that location~$j$ is searched, and thus that the cost $c_j$ needs to be paid. Hence, if we consider the objective function
\begin{align*}
\text{min}\, \sum_{i \in N} c_i \alpha_{i}
\end{align*}
and the constraints
\begin{align*}
& & \pi_i + \sum_{j \in N \setminus \{i\}} \pi_j(\mu_{\{i,j\}} + \delta_{ij}) &= \alpha_i & & \text{$\forall\, i \in N$}
\end{align*}
instead of Objective function~\eqref{mip:po_obj} and Constraints~\eqref{mip:po_y0}-\eqref{mip:po_yp} in the partial-order-based formulation for TCTP, then we obtain a valid formulation for TCSP.\@ Observe that also for TCSP, the assumption that~$n=mT$ is without loss of generality because we can add dummy locations with zero cost and zero probability.

\subsubsection{Incorporating precedence constraints}
Consider a generalization of TCTP or TCSP in which there are \emph{precedence constraints} defined by a strict partial order~$A \subset N \times N$ on~$N$, where~$(i,j) \in A$ reflects that component~$j \in N$ can only be tested or searched if component~$i \in N$ has been tested or searched before. Such constraints can easily be incorporated in our partial-order-based formulations by just adding a constraint $\delta_{ij} = 1$ for each $(i,j) \in A$.

By additionally requiring that~$m = 1$ and $T=n$, we obtain a compact mixed integer program for the precedence-constrained sequential testing problem. \cite{rostami2019sequential} describe a dynamic program and a branch-and-price algorithm for the precedence-constrained sequential testing problem for $n$-out-of-$n$ systems. Their branch-and-price approach also relies on a related linear-order-based formulation, but the authors use a Dantzig-Wolfe decomposition with exponentially many decision variables to linearize the objective.

\subsection{Assignment-based formulation}\label{sec:ab_form}
In the MIP formulations below, we use binary decision variables that assign the tests (or searches) to the time slots. Additionally, we include decision variables that keep track of the probability that the system's state is still unknown (or that the target is not found) at the beginning of a time slot.

\subsubsection{Testing variant}
For each test $j \in N$ and time slot~$t \in \{1,\ldots,T\}$, define a binary decision variable~$x_{jt}$ that equals 1 if test~$j$ is scheduled in time slot~$t$ and 0 otherwise, and a continuous decision variable~$y_{j}$ equal to the probability that test~$j$ will actually be performed. Next, for each~$j = 0, \ldots, n$ and~$t \in \{1,\ldots,T\}$, define a continuous decision variable~$z_{jt}$ equal to the probability that all tests scheduled before time~$t$ are successful, conditional on the information that all tests~$i > j$ in time slot~$t-1$ are successful. Hence,~$z_{nt}$ equals the probability that the tests scheduled in time slot~$t$ need to be performed. This leads to the following MIP formulation for TCTP.\@
\begin{subequations}
	\begin{align}
	& \text{min} & \sum_{j \in N} c_j y_j  \label{mip:ab_obj}\\
	&\text{s.t.} & \sum_{t=1}^T x_{jt} & =  1 & & \text{$\forall\, j \in N$} \label{mip:ab_test}\\		
	& & \sum_{j \in N} x_{jt} & \leq m & & \text{$\forall\, t = 1, \ldots, T$} \label{mip:ab_mach}\\
	& & y_j &\geq z_{nt} - 1 + \sum_{s = 1}^tx_{js} & & \text{$\forall\, j \in N$ and $t = 1,\ldots, T$} \label{mip:ab_y}\\
	& &  z_{j1} & = 1  & & \text{$\forall\, j =0,\ldots,n$} \label{mip:ab_z1}\\
	& & z_{0t} &= z_{n,t-1} & & \text{$\forall\, t = 2, \ldots, T$} \label{mip:ab_ztrans}\\
	& & z_{jt} &\geq z_{j-1,t} - x_{j,t-1} & & \text{$\forall\, j = 1, \ldots,n$ and $t = 2,\ldots, T$} \label{mip:ab_z}\\
	& & z_{jt} &\geq p_{j} z_{j-1,t} & & \text{$\forall\, j = 1, \ldots,n$ and $t = 2,\ldots, T$} \label{mip:ab_zp}\\
	& &  x_{jt} &  \in \{0,1\} & & \text{$\forall\, j \in N$ and $t = 1, \ldots, T$} \label{mip:ab_bin}
	\end{align}	
\end{subequations}

Objective function~\eqref{mip:ab_obj} combines the probability that a test needs to be performed with its cost. Constraints~\eqref{mip:ab_test}-\eqref{mip:ab_mach} together with the binary constraints enforce that every test is scheduled in exactly one time slot, and that at most~$m$ tests are scheduled within a time slot. Constraints~\eqref{mip:ab_y} together with the objective function and the fact that~$z_{nt}$ is non-increasing in~$t$ make sure that if test~$j$ is scheduled in or before time slot~$t$, it will be performed with a probability of at least~$z_{nt}$. Constraints~\eqref{mip:ab_z1}-\eqref{mip:ab_zp}, finally, enforce  that the probability~$z_{jt}$ equals~$p_j z_{j-1,t}$ if test~$j$ is scheduled in time slot~$t-1$, and~$z_{j-1,t}$ otherwise.

\subsubsection{Search variant}
For TCSP we adopt the same definition for~$x_{jt}$ and~$y_j$, but now we directly let~$z_t$ reflect the probability that the locations scheduled in time slot~$t \in \{1,\ldots, T\}$ need to be searched. That is,~$z_t$ equals the probability that the object is hidden in one of the locations searched at time~$t$ or later. Hence, if we consider the constraints
\begin{align*}
& & y_{j} &\geq z_{t} - 1 + \sum_{s=1}^tx_{js} & & \text{$\forall\, j \in N$ and $t = 1,\ldots, T$}\\
& & z_{t} &= \sum_{k=t}^T \sum_{j \in N} \pi_j x_{jk} & & \text{$\forall\, t = 2, \ldots, T$}
\end{align*}
instead of Constraints~\eqref{mip:ab_y}-\eqref{mip:ab_zp} in the assignment-based formulation for TCTP, then we obtain a valid formulation for the~TCSP.\@

\subsubsection{Incorporating a fixed cost per used time slot}
Consider a generalization of TCTP (or TCSP) in which there is a, possibly time-dependent, fixed cost~$\beta_t \in \mathbb{R}_{\geq 0}$ to test (or search) a batch of components (or locations) at time~$t \in \{1,\ldots, T\}$. This cost~$\beta_t$ could reflect a fixed start-up cost, but possibly also a penalty for not knowing the system's state at time~$t$. We can incorporate such fixed costs into the assignment-based formulation for TCTP by introducing for each time slot~$t \in \{1,\ldots, T\}$ a decision variable~$u_t$ that equals the probability that one or more tests will take place in time slot~$t$. We then add to the objective function~\eqref{mip:ab_obj} an additional summation $\sum_{t = 1}^T \beta_t u_{t}$, and we also include the constraints
\begin{equation*}
u_t \geq z_{nt} - 1 + \sum_{j \in N} x_{jt}
\end{equation*}
and~$u_t \geq 0$ for every~$t \in \{1,\ldots, T\}$. An analogous approach allows us to include fixed costs into the assignment-based formulation for TCSP as well. 

By setting~$m=T=n$ such that there is no deadline or constraint on the number of tests that can be performed in parallel, we obtain a compact exact mixed integer program for the batch-testing problem. \citet{daldal2016approximation} also describe an assignment-based formulation for this problem, but they approximate the objective function to obtain linearity.

\section{Local search heuristic}\label{sec:heur}
We now describe a local search procedure for TCTP and TCSP that exploits the fact that, for a given partition, we can determine an optimal schedule efficiently (Lemma~\ref{lem:struct_prop}). In particular, we represent a solution by means of an unordered partition~$\{S_1,\ldots, S_T\}$ of the set~$N$ into~$T$ sets with  at most~$m$ elements each, and translate this partition into a schedule by sequencing the sets in non-decreasing order of their cost-to-probability ratio.

The local search procedure for TCTP works as follows. Starting from a partition, a move consists of either the pairwise interchange between tests in different sets or the insertion of a test in another time slot in which less than~$m$ machines are being used. Moreover, after each interchange or insertion, the sets are again sorted in non-decreasing order of their cost-to-probability ratio. If there exists such a move that leads to a strict improvement, then we restart the procedure with this new partition. Otherwise, the partition defines a \emph{locally optimal testing sequence}~$\sigma^\text{ls} = (S^\text{ls}_1, \ldots, S^\text{ls}_T)$ and the procedure ends. The local search for TCSP works in a completely analogous fashion.

As illustrated by the following example, this local search procedure does not provide a constant-factor approximation guarantee for TCTP in general. The locality gap can in particular be~$\Omega(m)$, even if~$T=2$. The example does not seem to be adaptable to TCSP, as in this latter problem all probabilities add up to one (i.e., $\sum_{j \in N} \pi_j = 1$).

\begin{example}
	For given constants~$c, k, M \in \mathbb{N}$ all strictly greater than~1, define $p = (c+kM)^{-1/k}$. Consider a TCTP instance with~$m = 2k$ machines, deadline~$T = 2$, and the following~$n = 2k + 1$ tests:
	\begin{itemize}
		\item $I_1 \coloneqq \{1,\ldots, k\}$ with~$p_i = p$ and~$c_i = c-1$ for every~$i \in I_1$;
		\item $I_2 \coloneqq \{k+1,\ldots, 2k\}$ with~$p_i = 1$ and~$c_i = M$ for every~$i \in I_2$;
		\item $i^\star \coloneqq 2k+1$ with~$p_{i^\star} = 0$ and~$c_{i^\star} = c$.
	\end{itemize}
	For given values of~$c$ and~$k$, we will derive conditions for the parameter~$M$ such that the schedule~$\sigma^\text{ls} = (I_1, \{i^\star\} \cup I_2)$ is locally optimal and the schedule~$\sigma^\star = (\{i^\star\}, I_1\cup I_2)$ is globally optimal. Since, using $p = (c+kM)^{-1/k}$ and $m = 2k$,
	\begin{equation*}
	z(\sigma^\text{ls}) = k (c-1) + p^k \left(c + k M\right)  = k (c-1) + 1 = \frac{m}{2}(c-1) + 1,
	\end{equation*}
	whereas $z(\sigma^\star) = c$, this would then yield the~$\Omega(m)$ locality gap.
	
	Observe that  the potentially interesting moves from~$\sigma^\text{ls}$ are to either interchange a test in~$I_1$ with~$i^\star$ or to move a test in~$I_1$ to the second time slot. Interchanging a test in~$I_1$ with the test~$i^\star$ leads to a cost $(k-1)(c-1) + c = k(c-1) + 1 = z(\sigma^\text{ls})$ and, as such, it does not yield a strict improvement. Moving a test in~$I_1$ to the second time slot, in turn, leads to a cost $(k-1)(c-1) + p^{k-1} (c + c - 1 + kM)$. By substituting the definition of~$p$, we thus obtain that~$\sigma^\text{ls}$ is a local optimum if
	\begin{equation*}
	(k-1)(c-1) + \frac{2c - 1 + kM}{c+kM} (c+kM)^{1/k} \geq k (c-1) + 1,
	\end{equation*}
	which is equivalent to requiring that
	\begin{equation*}
	\frac{c^2+ckM}{2c - 1 + kM} (c+kM)^{-1/k} \leq 1.
	\end{equation*}
	Since the left-hand side of this latter inequality tends to zero as~$M$ increases, there exists a value $M \in \mathbb{N}_0$ such that this inequality is satisfied. For every~$k,c \in \mathbb{N}_0$, we can thus choose~$M$ sufficiently large such that~$\sigma^\text{ls}$ is a local optimum, which establishes the~$\Omega(m)$ locality gap.
\end{example}

\section{Computational experiments}\label{sec:comput}
This section reports on the computational performance of our proposed solution methods. After describing the dataset of test instances and the implementation details (Section~\ref{sec:implementation}), we discuss the performance of our mixed integer programming formulations (Section~\ref{sec:perf_mip}) as well as of our local search heuristic (Section~\ref{sec:perf_ls}).

\subsection{Instance description and implementation details}\label{sec:implementation}
A set of TCTP and TCSP instances for different values of~$m$ and~$T$, with $n = m T$, was randomly generated as follows. For each~$j \in N$, an integer cost~$c_j$ and weight~$w_j$ was drawn randomly, i.e., with uniform probability, from the interval~$[0, 10]$ and~$[0,1000]$, respectively. For TCSP, we then let~$\pi_j = w_j / w(N)$ for each~$j \in N$. For TCTP, in turn, a joint success probability~$q = \prod_{j \in N} p_j$ was first drawn randomly from the interval~$[0.01,0.30]$,~$[0.31,0.60]$, or~$[0.61,0.90]$ depending on the experimental setting. This probability was subsequently distributed randomly over the tests by setting $\log(p_j) = \log(q) w_j / w(N)$ for each~$j \in N$. By controlling the joint success probability, we avoid that the product~$\prod_{j \in N} p_j$ becomes very small as~$n$ grows, which could cause numerical issues. For each setting,~$10$ instances were generated, leading to~$30$ TCTP and~$10$ TCSP instances for each $(m,T)$-combination.

All our algorithms were implemented using the C++ programming language, compiled with Microsoft Visual C++ 14.0, and run using an Intel Core i7-4790 processor with 3.60 GHz CPU speed and 8 GB of RAM under a Windows 10 64-bit OS.\@ To solve the instances, we used the commercial solver IBM ILOG CPLEX 12.8 with a single thread and a time limit of~30 minutes. Lazy cuts were used to implement Constraints~\eqref{mip:po_transd} of the partial-order-based formulation, and the solution provided by the local search procedure of Section~\ref{sec:heur} was used as a warm start. Apart from this, all CPLEX parameters were set to their default values. 

To obtain a starting solution for each instance, our local search procedure was run with three different initializations; the resulting solutions were all provided to the solver. For TCTP, the construction of an initial partition was based on either the costs~$c_j$, the probabilities~$p_j$, or the ratios~$c_j / (1-p_j)$, respectively. In particular, all tests were first sorted in non-decreasing value of the corresponding quantity (where those tests~$j \in N$ with~$p_j = 1$ are put first in this order if~$c_j = 0$, and last otherwise), and an initial partition was subsequently obtained by filling the different time slots using the resulting order. For TCSP, an analogous approach based on the costs~$c_j$, probabilities~$\pi_j$, or ratios~$c_j / \pi_j$ was used.

In order to evaluate the computational performance of our formulations, we use the average CPU time in seconds, the number of solved instances within the 30-minute time limit, the average LP gap, and the average final gap. Here, the LP gap equals the percentage ratio between the value of the optimal solution and the lower bound provided by the linear program obtained by relaxing the integrality constraints in the corresponding MIP formulation. The final gap, in turn, equals the percentage by which the best found solution exceeds the current global lower bound when the 30-minute time limit is reached. The average CPU time and LP gap are computed using solved instances only, whereas the average final gap is computed using unsolved instances only. 
If all instances for a given $(m,T)$-combination remained unsolved by a method, then no instances with higher values for~$m$ and~$T$ were attempted with that method.

\subsection{Performance of mixed integer programming formulations}\label{sec:perf_mip}

Table~\ref{tab:comp_res} displays the computational performance of our~TCTP and~TCSP formulations. The first and most remarkable aspect emerging from this table is that TCSP appears to be considerably easier to solve than TCTP.\@ While we can solve almost all TCSP instances with $n = mT\leq 40$, regardless of the value for~$m$ and~$T$, we have to settle for a much smaller instance size $n\leq 16$ for TCTP.\@ Somewhat surprisingly, we also find that the assignment-based formulation performs better than the partial-order-based formulation for TCTP, whereas this is the other way round for TCSP.\@

\begin{table}\footnotesize
	\setlength\tabcolsep{4pt}
	\centering
	\renewcommand{\arraystretch}{0.6}
	\caption{Computational performance of partial-order-based and assignment-based formulations for TCTP and TCSP.\@ The averages for the CPU time (in seconds) and percentage LP gap (\%LP) are computed using solved instances only (\# out of 30 for TCTP and out of 10 for TCSP, with a 30-minute time limit). The average final gap (\%fin) is computed using unsolved instances only.} \label{tab:comp_res}
	\begin {tabular}{rrrrcrrrcrrrcrrrcr}%
	\toprule & & \multicolumn {8}{c}{Time-critical testing problem} & \multicolumn {8}{c}{Time-critical search problem}\\ \cmidrule (lr){3-10} \cmidrule (lr){11-18} & & \multicolumn {4}{c}{Partial order} & \multicolumn {4}{c}{Assignment} & \multicolumn {4}{c}{Partial order} & \multicolumn {4}{c}{Assignment} \\ \cmidrule (lr){3-6} \cmidrule (lr){7-10} \cmidrule (lr){11-14} \cmidrule (lr){15-18}$m$&$T$&\multicolumn {1}{c}{CPU}&\multicolumn {1}{r}{\#}&\multicolumn {1}{c}{\%LP}&\multicolumn {1}{c}{\%fin}&\multicolumn {1}{c}{CPU}&\multicolumn {1}{r}{\#}&\multicolumn {1}{c}{\%LP}&\multicolumn {1}{c}{\%fin}&\multicolumn {1}{c}{CPU}&\multicolumn {1}{r}{\#}&\multicolumn {1}{c}{\%LP}&\multicolumn {1}{c}{\%fin}&\multicolumn {1}{c}{CPU}&\multicolumn {1}{r}{\#}&\multicolumn {1}{c}{\%LP}&\multicolumn {1}{c}{\%fin}\\\midrule %
	\pgfutilensuremath {2}&\pgfutilensuremath {2}&\pgfutilensuremath {0.1}&\pgfutilensuremath {30}&\pgfutilensuremath {27.08}&&\pgfutilensuremath {0.1}&\pgfutilensuremath {30}&\pgfutilensuremath {24.45}&&\pgfutilensuremath {0.1}&\pgfutilensuremath {10}&\pgfutilensuremath {0.06}&&\pgfutilensuremath {0.0}&\pgfutilensuremath {10}&\pgfutilensuremath {17.22}&\\%
	\pgfutilensuremath {2}&\pgfutilensuremath {3}&\pgfutilensuremath {0.1}&\pgfutilensuremath {30}&\pgfutilensuremath {26.68}&&\pgfutilensuremath {0.1}&\pgfutilensuremath {30}&\pgfutilensuremath {55.48}&&\pgfutilensuremath {0.1}&\pgfutilensuremath {10}&\pgfutilensuremath {0.63}&&\pgfutilensuremath {0.1}&\pgfutilensuremath {10}&\pgfutilensuremath {38.80}&\\%
	\pgfutilensuremath {2}&\pgfutilensuremath {4}&\pgfutilensuremath {0.2}&\pgfutilensuremath {30}&\pgfutilensuremath {25.77}&&\pgfutilensuremath {0.2}&\pgfutilensuremath {30}&\pgfutilensuremath {67.01}&&\pgfutilensuremath {0.1}&\pgfutilensuremath {10}&\pgfutilensuremath {0.56}&&\pgfutilensuremath {0.1}&\pgfutilensuremath {10}&\pgfutilensuremath {51.24}&\\%
	\pgfutilensuremath {2}&\pgfutilensuremath {5}&\pgfutilensuremath {14.7}&\pgfutilensuremath {30}&\pgfutilensuremath {24.02}&&\pgfutilensuremath {4.5}&\pgfutilensuremath {30}&\pgfutilensuremath {76.34}&&\pgfutilensuremath {0.1}&\pgfutilensuremath {10}&\pgfutilensuremath {0.45}&&\pgfutilensuremath {0.4}&\pgfutilensuremath {10}&\pgfutilensuremath {60.45}&\\%
	\pgfutilensuremath {2}&\pgfutilensuremath {6}&\pgfutilensuremath {113.9}&\pgfutilensuremath {29}&\pgfutilensuremath {30.30}&\pgfutilensuremath {0.20}&\pgfutilensuremath {152.1}&\pgfutilensuremath {30}&\pgfutilensuremath {80.47}&&\pgfutilensuremath {0.1}&\pgfutilensuremath {10}&\pgfutilensuremath {0.50}&&\pgfutilensuremath {1.2}&\pgfutilensuremath {10}&\pgfutilensuremath {66.06}&\\%
	\pgfutilensuremath {2}&\pgfutilensuremath {7}&\pgfutilensuremath {439.3}&\pgfutilensuremath {13}&\pgfutilensuremath {26.76}&\pgfutilensuremath {1.55}&\pgfutilensuremath {843.3}&\pgfutilensuremath {7}&\pgfutilensuremath {96.72}&\pgfutilensuremath {18.70}&\pgfutilensuremath {0.1}&\pgfutilensuremath {10}&\pgfutilensuremath {0.49}&&\pgfutilensuremath {20.7}&\pgfutilensuremath {10}&\pgfutilensuremath {71.32}&\\%
	\pgfutilensuremath {2}&\pgfutilensuremath {8}&&\pgfutilensuremath {0}&&\pgfutilensuremath {6.10}&\pgfutilensuremath {777.5}&\pgfutilensuremath {1}&\pgfutilensuremath {87.44}&\pgfutilensuremath {44.37}&\pgfutilensuremath {0.1}&\pgfutilensuremath {10}&\pgfutilensuremath {0.29}&&\pgfutilensuremath {514.5}&\pgfutilensuremath {8}&\pgfutilensuremath {75.35}&\pgfutilensuremath {8.30}\\%
	\pgfutilensuremath {2}&\pgfutilensuremath {9}&&&&&&\pgfutilensuremath {0}&&\pgfutilensuremath {62.04}&\pgfutilensuremath {0.1}&\pgfutilensuremath {10}&\pgfutilensuremath {0.24}&&\pgfutilensuremath {207.2}&\pgfutilensuremath {4}&\pgfutilensuremath {77.98}&\pgfutilensuremath {15.48}\\%
	\pgfutilensuremath {2}&\pgfutilensuremath {10}&&&&&&&&&\pgfutilensuremath {0.1}&\pgfutilensuremath {10}&\pgfutilensuremath {0.27}&&\pgfutilensuremath {261.6}&\pgfutilensuremath {1}&\pgfutilensuremath {81.93}&\pgfutilensuremath {22.28}\\%
	\midrule \pgfutilensuremath {4}&\pgfutilensuremath {2}&\pgfutilensuremath {0.3}&\pgfutilensuremath {30}&\pgfutilensuremath {36.68}&&\pgfutilensuremath {0.0}&\pgfutilensuremath {30}&\pgfutilensuremath {28.46}&&\pgfutilensuremath {0.0}&\pgfutilensuremath {10}&\pgfutilensuremath {0.18}&&\pgfutilensuremath {0.0}&\pgfutilensuremath {10}&\pgfutilensuremath {19.44}&\\%
	\pgfutilensuremath {4}&\pgfutilensuremath {3}&\pgfutilensuremath {271.1}&\pgfutilensuremath {23}&\pgfutilensuremath {35.20}&\pgfutilensuremath {1.11}&\pgfutilensuremath {0.5}&\pgfutilensuremath {30}&\pgfutilensuremath {52.62}&&\pgfutilensuremath {0.1}&\pgfutilensuremath {10}&\pgfutilensuremath {0.47}&&\pgfutilensuremath {0.1}&\pgfutilensuremath {10}&\pgfutilensuremath {39.75}&\\%
	\pgfutilensuremath {4}&\pgfutilensuremath {4}&&\pgfutilensuremath {0}&&\pgfutilensuremath {5.50}&\pgfutilensuremath {237.3}&\pgfutilensuremath {30}&\pgfutilensuremath {69.04}&&\pgfutilensuremath {0.2}&\pgfutilensuremath {10}&\pgfutilensuremath {1.05}&&\pgfutilensuremath {0.7}&\pgfutilensuremath {10}&\pgfutilensuremath {51.05}&\\%
	\pgfutilensuremath {4}&\pgfutilensuremath {5}&&&&&&\pgfutilensuremath {0}&&\pgfutilensuremath {31.12}&\pgfutilensuremath {0.4}&\pgfutilensuremath {10}&\pgfutilensuremath {0.54}&&\pgfutilensuremath {349.1}&\pgfutilensuremath {9}&\pgfutilensuremath {60.57}&\pgfutilensuremath {3.68}\\%
	\pgfutilensuremath {4}&\pgfutilensuremath {6}&&&&&&&&&\pgfutilensuremath {1.0}&\pgfutilensuremath {10}&\pgfutilensuremath {0.43}&&&\pgfutilensuremath {0}&&\pgfutilensuremath {12.80}\\%
	\pgfutilensuremath {4}&\pgfutilensuremath {7}&&&&&&&&&\pgfutilensuremath {3.1}&\pgfutilensuremath {10}&\pgfutilensuremath {0.39}&&&&&\\%
	\pgfutilensuremath {4}&\pgfutilensuremath {8}&&&&&&&&&\pgfutilensuremath {140.9}&\pgfutilensuremath {10}&\pgfutilensuremath {0.22}&&&&&\\%
	\pgfutilensuremath {4}&\pgfutilensuremath {9}&&&&&&&&&\pgfutilensuremath {224.7}&\pgfutilensuremath {10}&\pgfutilensuremath {0.59}&&&&&\\%
	\pgfutilensuremath {4}&\pgfutilensuremath {10}&&&&&&&&&\pgfutilensuremath {213.0}&\pgfutilensuremath {10}&\pgfutilensuremath {0.35}&&&&&\\%
	\midrule \pgfutilensuremath {6}&\pgfutilensuremath {2}&\pgfutilensuremath {277.4}&\pgfutilensuremath {27}&\pgfutilensuremath {36.41}&\pgfutilensuremath {0.47}&\pgfutilensuremath {0.1}&\pgfutilensuremath {30}&\pgfutilensuremath {32.39}&&\pgfutilensuremath {0.1}&\pgfutilensuremath {10}&\pgfutilensuremath {0.01}&&\pgfutilensuremath {0.1}&\pgfutilensuremath {10}&\pgfutilensuremath {21.09}&\\%
	\pgfutilensuremath {6}&\pgfutilensuremath {3}&&\pgfutilensuremath {0}&&\pgfutilensuremath {11.41}&\pgfutilensuremath {28.8}&\pgfutilensuremath {30}&\pgfutilensuremath {55.26}&&\pgfutilensuremath {0.2}&\pgfutilensuremath {10}&\pgfutilensuremath {0.29}&&\pgfutilensuremath {0.3}&\pgfutilensuremath {10}&\pgfutilensuremath {39.46}&\\%
	\pgfutilensuremath {6}&\pgfutilensuremath {4}&&&&&&\pgfutilensuremath {0}&&\pgfutilensuremath {24.19}&\pgfutilensuremath {0.8}&\pgfutilensuremath {10}&\pgfutilensuremath {0.39}&&\pgfutilensuremath {61.5}&\pgfutilensuremath {10}&\pgfutilensuremath {52.42}&\\%
	\pgfutilensuremath {6}&\pgfutilensuremath {5}&&&&&&&&&\pgfutilensuremath {3.4}&\pgfutilensuremath {10}&\pgfutilensuremath {0.38}&&\pgfutilensuremath {883.9}&\pgfutilensuremath {1}&\pgfutilensuremath {61.27}&\pgfutilensuremath {11.46}\\%
	\pgfutilensuremath {6}&\pgfutilensuremath {6}&&&&&&&&&\pgfutilensuremath {140.8}&\pgfutilensuremath {10}&\pgfutilensuremath {0.37}&&&\pgfutilensuremath {0}&&\pgfutilensuremath {30.32}\\%
	\pgfutilensuremath {6}&\pgfutilensuremath {7}&&&&&&&&&\pgfutilensuremath {107.4}&\pgfutilensuremath {5}&\pgfutilensuremath {0.27}&\pgfutilensuremath {0.30}&&&&\\%
	\pgfutilensuremath {6}&\pgfutilensuremath {8}&&&&&&&&&\pgfutilensuremath {157.6}&\pgfutilensuremath {5}&\pgfutilensuremath {0.19}&\pgfutilensuremath {0.21}&&&&\\%
	\pgfutilensuremath {6}&\pgfutilensuremath {9}&&&&&&&&&\pgfutilensuremath {220.1}&\pgfutilensuremath {2}&\pgfutilensuremath {0.16}&\pgfutilensuremath {0.24}&&&&\\%
	\pgfutilensuremath {6}&\pgfutilensuremath {10}&&&&&&&&&&\pgfutilensuremath {0}&&\pgfutilensuremath {0.27}&&&&\\%
	\midrule \pgfutilensuremath {8}&\pgfutilensuremath {2}&\pgfutilensuremath {921.4}&\pgfutilensuremath {7}&\pgfutilensuremath {72.61}&\pgfutilensuremath {4.02}&\pgfutilensuremath {0.1}&\pgfutilensuremath {30}&\pgfutilensuremath {29.45}&&\pgfutilensuremath {0.1}&\pgfutilensuremath {10}&\pgfutilensuremath {0.00}&&\pgfutilensuremath {0.1}&\pgfutilensuremath {10}&\pgfutilensuremath {16.60}&\\%
	\pgfutilensuremath {8}&\pgfutilensuremath {3}&&\pgfutilensuremath {0}&&\pgfutilensuremath {14.44}&\pgfutilensuremath {574.7}&\pgfutilensuremath {25}&\pgfutilensuremath {61.52}&\pgfutilensuremath {1.94}&\pgfutilensuremath {0.6}&\pgfutilensuremath {10}&\pgfutilensuremath {0.26}&&\pgfutilensuremath {0.6}&\pgfutilensuremath {10}&\pgfutilensuremath {36.87}&\\%
	\pgfutilensuremath {8}&\pgfutilensuremath {4}&&&&&&\pgfutilensuremath {0}&&\pgfutilensuremath {38.08}&\pgfutilensuremath {8.0}&\pgfutilensuremath {10}&\pgfutilensuremath {0.29}&&\pgfutilensuremath {637.9}&\pgfutilensuremath {3}&\pgfutilensuremath {52.43}&\pgfutilensuremath {5.93}\\%
	\pgfutilensuremath {8}&\pgfutilensuremath {5}&&&&&&&&&\pgfutilensuremath {66.6}&\pgfutilensuremath {10}&\pgfutilensuremath {0.31}&&&\pgfutilensuremath {0}&&\pgfutilensuremath {23.80}\\%
	\pgfutilensuremath {8}&\pgfutilensuremath {6}&&&&&&&&&\pgfutilensuremath {448.6}&\pgfutilensuremath {7}&\pgfutilensuremath {0.27}&\pgfutilensuremath {0.17}&&&&\\%
	\pgfutilensuremath {8}&\pgfutilensuremath {7}&&&&&&&&&\pgfutilensuremath {458.1}&\pgfutilensuremath {6}&\pgfutilensuremath {0.17}&\pgfutilensuremath {0.30}&&&&\\%
	\pgfutilensuremath {8}&\pgfutilensuremath {8}&&&&&&&&&\pgfutilensuremath {108.4}&\pgfutilensuremath {1}&\pgfutilensuremath {0.09}&\pgfutilensuremath {0.15}&&&&\\%
	\pgfutilensuremath {8}&\pgfutilensuremath {9}&&&&&&&&&&\pgfutilensuremath {0}&&\pgfutilensuremath {0.34}&&&&\\%
	\midrule \pgfutilensuremath {10}&\pgfutilensuremath {2}&&\pgfutilensuremath {0}&&\pgfutilensuremath {13.96}&\pgfutilensuremath {0.1}&\pgfutilensuremath {30}&\pgfutilensuremath {33.39}&&\pgfutilensuremath {0.1}&\pgfutilensuremath {10}&\pgfutilensuremath {0.02}&&\pgfutilensuremath {0.1}&\pgfutilensuremath {10}&\pgfutilensuremath {19.60}&\\%
	\pgfutilensuremath {10}&\pgfutilensuremath {3}&&&&&\pgfutilensuremath {1{,}210.3}&\pgfutilensuremath {1}&\pgfutilensuremath {60.97}&\pgfutilensuremath {13.10}&\pgfutilensuremath {1.9}&\pgfutilensuremath {10}&\pgfutilensuremath {0.26}&&\pgfutilensuremath {32.9}&\pgfutilensuremath {10}&\pgfutilensuremath {37.77}&\\%
	\pgfutilensuremath {10}&\pgfutilensuremath {4}&&&&&&\pgfutilensuremath {0}&&\pgfutilensuremath {43.30}&\pgfutilensuremath {22.6}&\pgfutilensuremath {10}&\pgfutilensuremath {0.25}&&&\pgfutilensuremath {0}&&\pgfutilensuremath {9.56}\\%
	\pgfutilensuremath {10}&\pgfutilensuremath {5}&&&&&&&&&\pgfutilensuremath {284.7}&\pgfutilensuremath {10}&\pgfutilensuremath {0.20}&&&&&\\%
	\pgfutilensuremath {10}&\pgfutilensuremath {6}&&&&&&&&&\pgfutilensuremath {926.9}&\pgfutilensuremath {5}&\pgfutilensuremath {0.25}&\pgfutilensuremath {0.17}&&&&\\%
	\pgfutilensuremath {10}&\pgfutilensuremath {7}&&&&&&&&&&\pgfutilensuremath {0}&&\pgfutilensuremath {0.40}&&&&\\\bottomrule %
	\end {tabular}%
\end{table}

The partial-order-based formulation for TCTP performs quite poorly when both~$T$ and~$m$ exceed~$2$. This is somewhat surprising, given that its LP gap is significantly smaller than the one of the assignment-based formulation. One possible explanation is that the partial-order-based formulation, with $O(n^3)$ constraints, is relatively large. This makes the formulation stronger, but it also increases the time needed to solve the LP relaxation in every node of the branch-and-bound tree. For both formulations, the performance seems to be more sensitive to the value of the deadline~$T$ than to the number $m$ of testers. In fact, while all instances with $m=4$ and $T=4$ can be solved to optimality, there is only one solved instance with $m=2$ and $T=8$.

For TCSP, in contrast, the partial-order-based formulation clearly outperforms the assignment-based formulation. This superior performance is most likely explained by the remarkably small LP gap, which is several orders of magnitude below the gap of the assignment-based formulation.   As such, our results are consistent with the earlier finding that partial-order-based formulations typically perform well for scheduling problems with a total weighted completion time objective function (see for instance \citeauthor{Potts1980},~\citeyear{Potts1980} or \citeauthor{keha2009mixed},~\citeyear{keha2009mixed}).

\begin{table}\footnotesize
	\setlength\tabcolsep{3pt}
	\centering
	\renewcommand{\arraystretch}{0.6}
	\caption{Computational performance for the partial-order based formulation for larger TCSP instances.} \label{tab:comp_po_detail}
	\begin{subtable}{\textwidth}
		\caption{Partial-order based formulation for TCSP with larger values for the deadline~$T$.}\label{tab:largeT}
		\centering
		\begin {tabular}{rrrcrrrcr}%
		\toprule & \multicolumn {4}{c}{$m = 2$} & \multicolumn {4}{c}{$m = 4$} \\ \cmidrule (lr){2-5} \cmidrule (lr){6-9}$T$&\multicolumn {1}{c}{CPU}&\multicolumn {1}{r}{\#}&\multicolumn {1}{c}{\%LP}&\multicolumn {1}{c}{\%fin}&\multicolumn {1}{c}{CPU}&\multicolumn {1}{r}{\#}&\multicolumn {1}{c}{\%LP}&\multicolumn {1}{c}{\%fin}\\\midrule %
		\pgfutilensuremath {15}&\pgfutilensuremath {0.5}&\pgfutilensuremath {10}&\pgfutilensuremath {0.20}&&\pgfutilensuremath {1{,}532.9}&\pgfutilensuremath {1}&\pgfutilensuremath {0.12}&\pgfutilensuremath {0.29}\\%
		\pgfutilensuremath {20}&\pgfutilensuremath {25.6}&\pgfutilensuremath {10}&\pgfutilensuremath {0.14}&&&\pgfutilensuremath {0}&&\pgfutilensuremath {0.26}\\%
		\pgfutilensuremath {25}&\pgfutilensuremath {48.8}&\pgfutilensuremath {10}&\pgfutilensuremath {0.09}&&&&&\\%
		\pgfutilensuremath {30}&\pgfutilensuremath {173.1}&\pgfutilensuremath {10}&\pgfutilensuremath {0.06}&&&&&\\%
		\pgfutilensuremath {35}&\pgfutilensuremath {695.5}&\pgfutilensuremath {5}&\pgfutilensuremath {0.05}&\pgfutilensuremath {0.06}&&&&\\%
		\pgfutilensuremath {40}&\pgfutilensuremath {99.7}&\pgfutilensuremath {1}&\pgfutilensuremath {0.03}&\pgfutilensuremath {0.05}&&&&\\%
		\pgfutilensuremath {45}&\pgfutilensuremath {890.2}&\pgfutilensuremath {2}&\pgfutilensuremath {0.02}&\pgfutilensuremath {0.08}&&&&\\%
		\pgfutilensuremath {50}&&\pgfutilensuremath {0}&&\pgfutilensuremath {0.09}&&&&\\\bottomrule %
		\end {tabular}%
	\end{subtable}
	\vspace{2ex}
	
	\begin{subtable}{\textwidth}
		\caption{Partial-order based formulation for TCSP with larger values for the number of searchers~$m$.}\label{tab:largem}
		\centering
		\begin {tabular}{rrrcrrrcrrrcrrrcr}%
		\toprule & \multicolumn {4}{c}{$T = 2$} & \multicolumn {4}{c}{$T = 3$} & \multicolumn {4}{c}{$T = 4$} & \multicolumn {4}{c}{$T = 5$}\\ \cmidrule (lr){2-5} \cmidrule (lr){6-9} \cmidrule (lr){10-13} \cmidrule (lr){14-17}$m$&\multicolumn {1}{c}{CPU}&\multicolumn {1}{r}{\#}&\multicolumn {1}{c}{\%LP}&\multicolumn {1}{c}{\%fin}&\multicolumn {1}{c}{CPU}&\multicolumn {1}{r}{\#}&\multicolumn {1}{c}{\%LP}&\multicolumn {1}{c}{\%fin}&\multicolumn {1}{c}{CPU}&\multicolumn {1}{r}{\#}&\multicolumn {1}{c}{\%LP}&\multicolumn {1}{c}{\%fin}&\multicolumn {1}{c}{CPU}&\multicolumn {1}{r}{\#}&\multicolumn {1}{c}{\%LP}&\multicolumn {1}{c}{\%fin}\\\midrule %
		\pgfutilensuremath {12}&\pgfutilensuremath {0.5}&\pgfutilensuremath {10}&\pgfutilensuremath {0.01}&&\pgfutilensuremath {4.8}&\pgfutilensuremath {10}&\pgfutilensuremath {0.19}&&\pgfutilensuremath {246.1}&\pgfutilensuremath {10}&\pgfutilensuremath {0.31}&&\pgfutilensuremath {113.3}&\pgfutilensuremath {4}&\pgfutilensuremath {0.09}&\pgfutilensuremath {0.26}\\%
		\pgfutilensuremath {14}&\pgfutilensuremath {0.6}&\pgfutilensuremath {10}&\pgfutilensuremath {0.02}&&\pgfutilensuremath {6.1}&\pgfutilensuremath {10}&\pgfutilensuremath {0.06}&&\pgfutilensuremath {370.2}&\pgfutilensuremath {10}&\pgfutilensuremath {0.19}&&\pgfutilensuremath {612.7}&\pgfutilensuremath {4}&\pgfutilensuremath {0.04}&\pgfutilensuremath {0.32}\\%
		\pgfutilensuremath {16}&\pgfutilensuremath {0.4}&\pgfutilensuremath {10}&\pgfutilensuremath {0.00}&&\pgfutilensuremath {12.1}&\pgfutilensuremath {10}&\pgfutilensuremath {0.05}&&\pgfutilensuremath {472.6}&\pgfutilensuremath {7}&\pgfutilensuremath {0.09}&\pgfutilensuremath {0.28}&\pgfutilensuremath {1{,}280.4}&\pgfutilensuremath {2}&\pgfutilensuremath {0.04}&\pgfutilensuremath {0.53}\\%
		\pgfutilensuremath {18}&\pgfutilensuremath {2.0}&\pgfutilensuremath {10}&\pgfutilensuremath {0.01}&&\pgfutilensuremath {26.0}&\pgfutilensuremath {10}&\pgfutilensuremath {0.01}&&\pgfutilensuremath {713.2}&\pgfutilensuremath {5}&\pgfutilensuremath {0.07}&\pgfutilensuremath {0.12}&\pgfutilensuremath {205.2}&\pgfutilensuremath {2}&\pgfutilensuremath {0.00}&\pgfutilensuremath {0.72}\\%
		\pgfutilensuremath {20}&\pgfutilensuremath {1.5}&\pgfutilensuremath {10}&\pgfutilensuremath {0.00}&&\pgfutilensuremath {97.1}&\pgfutilensuremath {10}&\pgfutilensuremath {0.05}&&\pgfutilensuremath {757.0}&\pgfutilensuremath {7}&\pgfutilensuremath {0.02}&\pgfutilensuremath {0.16}&&\pgfutilensuremath {0}&&\pgfutilensuremath {0.78}\\%
		\pgfutilensuremath {22}&\pgfutilensuremath {1.8}&\pgfutilensuremath {10}&\pgfutilensuremath {0.00}&&\pgfutilensuremath {246.4}&\pgfutilensuremath {10}&\pgfutilensuremath {0.07}&&\pgfutilensuremath {1{,}015.8}&\pgfutilensuremath {2}&\pgfutilensuremath {0.09}&\pgfutilensuremath {0.70}&&&&\\%
		\pgfutilensuremath {24}&\pgfutilensuremath {3.7}&\pgfutilensuremath {10}&\pgfutilensuremath {0.03}&&\pgfutilensuremath {457.3}&\pgfutilensuremath {10}&\pgfutilensuremath {0.08}&&\pgfutilensuremath {567.2}&\pgfutilensuremath {2}&\pgfutilensuremath {0.02}&\pgfutilensuremath {0.45}&&&&\\%
		\pgfutilensuremath {26}&\pgfutilensuremath {9.3}&\pgfutilensuremath {10}&\pgfutilensuremath {0.00}&&\pgfutilensuremath {395.8}&\pgfutilensuremath {5}&\pgfutilensuremath {0.04}&\pgfutilensuremath {0.31}&&\pgfutilensuremath {0}&&\pgfutilensuremath {0.86}&&&&\\%
		\pgfutilensuremath {28}&\pgfutilensuremath {18.9}&\pgfutilensuremath {10}&\pgfutilensuremath {0.00}&&\pgfutilensuremath {497.8}&\pgfutilensuremath {5}&\pgfutilensuremath {0.02}&\pgfutilensuremath {0.53}&&\pgfutilensuremath {0}&&\pgfutilensuremath {0.85}&&&&\\%
		\pgfutilensuremath {30}&\pgfutilensuremath {8.4}&\pgfutilensuremath {10}&\pgfutilensuremath {0.00}&&\pgfutilensuremath {1{,}184.6}&\pgfutilensuremath {7}&\pgfutilensuremath {0.01}&\pgfutilensuremath {0.89}&&\pgfutilensuremath {0}&&\pgfutilensuremath {0.89}&&&&\\\bottomrule %
		\end {tabular}%
	\end{subtable}
\end{table}

Table~\ref{tab:comp_po_detail} reports on the computational performance for the partial-order-based formulation for larger TCSP instances. It shows that the time limit appears to become binding only when $n = mT$ reaches~$60$. When comparing Tables~\ref{tab:largeT} and~\ref{tab:largem}, we observe that, similarly to TCTP, instances with low~$T$ and high~$m$ seem to be easier to solve than instances with high~$T$ and low~$m$ in general. For the special case where~$m = 2$, however, the partial-order-based formulation can still solve instances with a fairly high value for~$T$. Observe in this respect that, although we have established in Section~\ref{sec:compl} that \mbox{$m$-TCSP} is strongly NP-hard if~$m \geq 3$, the complexity status of TCSP in the special case where~$m = 2$ is still open.

\begin{table}\footnotesize
	\setlength\tabcolsep{2.9pt}
	\centering
	\renewcommand{\arraystretch}{0.6}
	\caption{Computational performance of the TCTP formulations for different values of~$q = \prod_{j\in N} p_j$.} \label{tab:comp_q_po}
	\begin {tabular}{rrrrcrrcrrcrrcrrcrrc}%
	\toprule & & \multicolumn {9}{c}{Partial-order-based formulation} & \multicolumn {9}{c}{Assignment-based formulation} \\ \cmidrule (lr){3-11} \cmidrule (lr){12-20} & & \multicolumn {3}{c}{$q \in [0.01, 0.30]$} & \multicolumn {3}{c}{$q \in [0.31, 0.60]$} & \multicolumn {3}{c}{$q \in [0.61, 0.90]$} & \multicolumn {3}{c}{$q \in [0.01, 0.30]$} & \multicolumn {3}{c}{$q \in [0.31, 0.60]$} & \multicolumn {3}{c}{$q \in [0.61, 0.90]$} \\ \cmidrule (lr){3-5} \cmidrule (lr){6-8} \cmidrule (lr){9-11} \cmidrule (lr){12-14} \cmidrule (lr){15-17} \cmidrule (lr){18-20}$m$&$T$&\multicolumn {1}{c}{CPU}&\multicolumn {1}{r}{\#}&\multicolumn {1}{c}{\%LP}&\multicolumn {1}{c}{CPU}&\multicolumn {1}{r}{\#}&\multicolumn {1}{c}{\%LP}&\multicolumn {1}{c}{CPU}&\multicolumn {1}{r}{\#}&\multicolumn {1}{c}{\%LP}&\multicolumn {1}{c}{CPU}&\multicolumn {1}{r}{\#}&\multicolumn {1}{c}{\%LP}&\multicolumn {1}{c}{CPU}&\multicolumn {1}{r}{\#}&\multicolumn {1}{c}{\%LP}&\multicolumn {1}{c}{CPU}&\multicolumn {1}{r}{\#}&\multicolumn {1}{c}{\%LP}\\\midrule %
	\pgfutilensuremath {2}&\pgfutilensuremath {2}&\pgfutilensuremath {0.1}&\pgfutilensuremath {10}&\pgfutilensuremath {49.31}&\pgfutilensuremath {0.1}&\pgfutilensuremath {10}&\pgfutilensuremath {23.96}&\pgfutilensuremath {0.1}&\pgfutilensuremath {10}&\pgfutilensuremath {7.95}&\pgfutilensuremath {0.1}&\pgfutilensuremath {10}&\pgfutilensuremath {27.59}&\pgfutilensuremath {0.1}&\pgfutilensuremath {10}&\pgfutilensuremath {32.04}&\pgfutilensuremath {0.1}&\pgfutilensuremath {10}&\pgfutilensuremath {13.73}\\%
	\pgfutilensuremath {2}&\pgfutilensuremath {3}&\pgfutilensuremath {0.1}&\pgfutilensuremath {10}&\pgfutilensuremath {48.84}&\pgfutilensuremath {0.1}&\pgfutilensuremath {10}&\pgfutilensuremath {23.63}&\pgfutilensuremath {0.1}&\pgfutilensuremath {10}&\pgfutilensuremath {7.58}&\pgfutilensuremath {0.1}&\pgfutilensuremath {10}&\pgfutilensuremath {69.31}&\pgfutilensuremath {0.1}&\pgfutilensuremath {10}&\pgfutilensuremath {65.37}&\pgfutilensuremath {0.1}&\pgfutilensuremath {10}&\pgfutilensuremath {31.77}\\%
	\pgfutilensuremath {2}&\pgfutilensuremath {4}&\pgfutilensuremath {0.3}&\pgfutilensuremath {10}&\pgfutilensuremath {48.58}&\pgfutilensuremath {0.2}&\pgfutilensuremath {10}&\pgfutilensuremath {22.19}&\pgfutilensuremath {0.1}&\pgfutilensuremath {10}&\pgfutilensuremath {6.54}&\pgfutilensuremath {0.1}&\pgfutilensuremath {10}&\pgfutilensuremath {76.59}&\pgfutilensuremath {0.2}&\pgfutilensuremath {10}&\pgfutilensuremath {79.36}&\pgfutilensuremath {0.2}&\pgfutilensuremath {10}&\pgfutilensuremath {45.08}\\%
	\pgfutilensuremath {2}&\pgfutilensuremath {5}&\pgfutilensuremath {12.3}&\pgfutilensuremath {10}&\pgfutilensuremath {46.04}&\pgfutilensuremath {10.5}&\pgfutilensuremath {10}&\pgfutilensuremath {19.03}&\pgfutilensuremath {21.3}&\pgfutilensuremath {10}&\pgfutilensuremath {6.98}&\pgfutilensuremath {1.9}&\pgfutilensuremath {10}&\pgfutilensuremath {86.93}&\pgfutilensuremath {5.3}&\pgfutilensuremath {10}&\pgfutilensuremath {86.55}&\pgfutilensuremath {6.3}&\pgfutilensuremath {10}&\pgfutilensuremath {55.53}\\%
	\pgfutilensuremath {2}&\pgfutilensuremath {6}&\pgfutilensuremath {103.7}&\pgfutilensuremath {10}&\pgfutilensuremath {61.86}&\pgfutilensuremath {78.6}&\pgfutilensuremath {10}&\pgfutilensuremath {20.38}&\pgfutilensuremath {164.5}&\pgfutilensuremath {9}&\pgfutilensuremath {6.24}&\pgfutilensuremath {18.1}&\pgfutilensuremath {10}&\pgfutilensuremath {87.51}&\pgfutilensuremath {155.9}&\pgfutilensuremath {10}&\pgfutilensuremath {92.25}&\pgfutilensuremath {282.4}&\pgfutilensuremath {10}&\pgfutilensuremath {61.64}\\%
	\pgfutilensuremath {2}&\pgfutilensuremath {7}&\pgfutilensuremath {754.6}&\pgfutilensuremath {5}&\pgfutilensuremath {50.86}&\pgfutilensuremath {312.9}&\pgfutilensuremath {5}&\pgfutilensuremath {15.05}&\pgfutilensuremath {124.6}&\pgfutilensuremath {3}&\pgfutilensuremath {6.09}&\pgfutilensuremath {843.3}&\pgfutilensuremath {7}&\pgfutilensuremath {96.72}&&&&&&\\%
	\pgfutilensuremath {2}&\pgfutilensuremath {8}&&&&&&&&&&\pgfutilensuremath {777.5}&\pgfutilensuremath {1}&\pgfutilensuremath {87.44}&&&&&&\\%
	\midrule \pgfutilensuremath {4}&\pgfutilensuremath {2}&\pgfutilensuremath {0.3}&\pgfutilensuremath {10}&\pgfutilensuremath {68.49}&\pgfutilensuremath {0.2}&\pgfutilensuremath {10}&\pgfutilensuremath {29.67}&\pgfutilensuremath {0.3}&\pgfutilensuremath {10}&\pgfutilensuremath {11.88}&\pgfutilensuremath {0.0}&\pgfutilensuremath {10}&\pgfutilensuremath {35.98}&\pgfutilensuremath {0.0}&\pgfutilensuremath {10}&\pgfutilensuremath {34.35}&\pgfutilensuremath {0.0}&\pgfutilensuremath {10}&\pgfutilensuremath {15.05}\\%
	\pgfutilensuremath {4}&\pgfutilensuremath {3}&\pgfutilensuremath {139.9}&\pgfutilensuremath {10}&\pgfutilensuremath {60.31}&\pgfutilensuremath {308.4}&\pgfutilensuremath {6}&\pgfutilensuremath {23.19}&\pgfutilensuremath {426.5}&\pgfutilensuremath {7}&\pgfutilensuremath {9.65}&\pgfutilensuremath {0.4}&\pgfutilensuremath {10}&\pgfutilensuremath {67.25}&\pgfutilensuremath {0.5}&\pgfutilensuremath {10}&\pgfutilensuremath {61.49}&\pgfutilensuremath {0.7}&\pgfutilensuremath {10}&\pgfutilensuremath {29.13}\\%
	\pgfutilensuremath {4}&\pgfutilensuremath {4}&&&&&&&&&&\pgfutilensuremath {94.4}&\pgfutilensuremath {10}&\pgfutilensuremath {78.99}&\pgfutilensuremath {217.8}&\pgfutilensuremath {10}&\pgfutilensuremath {79.94}&\pgfutilensuremath {399.9}&\pgfutilensuremath {10}&\pgfutilensuremath {48.21}\\%
	\midrule \pgfutilensuremath {6}&\pgfutilensuremath {2}&\pgfutilensuremath {49.9}&\pgfutilensuremath {10}&\pgfutilensuremath {56.96}&\pgfutilensuremath {295.7}&\pgfutilensuremath {9}&\pgfutilensuremath {33.40}&\pgfutilensuremath {541.0}&\pgfutilensuremath {8}&\pgfutilensuremath {14.09}&\pgfutilensuremath {0.1}&\pgfutilensuremath {10}&\pgfutilensuremath {45.78}&\pgfutilensuremath {0.1}&\pgfutilensuremath {10}&\pgfutilensuremath {35.94}&\pgfutilensuremath {0.1}&\pgfutilensuremath {10}&\pgfutilensuremath {15.45}\\%
	\pgfutilensuremath {6}&\pgfutilensuremath {3}&&&&&&&&&&\pgfutilensuremath {13.9}&\pgfutilensuremath {10}&\pgfutilensuremath {67.44}&\pgfutilensuremath {27.5}&\pgfutilensuremath {10}&\pgfutilensuremath {65.65}&\pgfutilensuremath {44.9}&\pgfutilensuremath {10}&\pgfutilensuremath {32.68}\\%
	\midrule \pgfutilensuremath {8}&\pgfutilensuremath {2}&\pgfutilensuremath {789.4}&\pgfutilensuremath {6}&\pgfutilensuremath {78.49}&\pgfutilensuremath {1{,}713.0}&\pgfutilensuremath {1}&\pgfutilensuremath {37.33}&&&&\pgfutilensuremath {0.1}&\pgfutilensuremath {10}&\pgfutilensuremath {40.08}&\pgfutilensuremath {0.1}&\pgfutilensuremath {10}&\pgfutilensuremath {32.33}&\pgfutilensuremath {0.1}&\pgfutilensuremath {10}&\pgfutilensuremath {15.94}\\%
	\pgfutilensuremath {8}&\pgfutilensuremath {3}&&&&&&&&&&\pgfutilensuremath {417.1}&\pgfutilensuremath {10}&\pgfutilensuremath {71.63}&\pgfutilensuremath {602.7}&\pgfutilensuremath {9}&\pgfutilensuremath {65.86}&\pgfutilensuremath {795.5}&\pgfutilensuremath {6}&\pgfutilensuremath {38.15}\\%
	\midrule \pgfutilensuremath {10}&\pgfutilensuremath {2}&&&&&&&&&&\pgfutilensuremath {0.1}&\pgfutilensuremath {10}&\pgfutilensuremath {47.05}&\pgfutilensuremath {0.1}&\pgfutilensuremath {10}&\pgfutilensuremath {39.88}&\pgfutilensuremath {0.1}&\pgfutilensuremath {10}&\pgfutilensuremath {13.22}\\%
	\pgfutilensuremath {10}&\pgfutilensuremath {3}&&&&&&&&&&\pgfutilensuremath {1{,}210.3}&\pgfutilensuremath {1}&\pgfutilensuremath {60.97}&&&&&&\\\bottomrule %
	\end {tabular}%
\end{table}

Table \ref{tab:comp_q_po} further details the computational performance for TCTP by disaggregating the results on the basis of the success probability~$q$. Here, instances with lower values of~$q$ seem to be relatively more tractable. To explain this behavior, recall that a smaller~$q$ results in a wider range for the values that the objective function can attain across the various solutions (Lemma~\ref{lem:trivial_approx}). This wider range might then be exploited by the solver to differentiate more easily among different solutions. The effect of~$q$ on the LP gap is also consistent with Lemma~\ref{lem:trivial_approx}.

	\begin{table}\footnotesize
		\centering
		\renewcommand{\arraystretch}{0.6}
		\caption{CPU time in seconds for the assignment-based formulation and the dynamic program (DP) of Section~\ref{sec:DP}, for TCTP and TCSP instances with $T = 2$ and $m = 10$, depending on the range from which the costs~$c_j$ are drawn.} \label{tab:comp_res_DP}
		\begin {tabular}{lrrrr}%
		\toprule & \multicolumn {2}{c}{TCTP} & \multicolumn {2}{c}{TCSP}\\ \cmidrule (lr){2-3} \cmidrule (lr){4-5}Interval for $c_j$&Assignment&DP&Assignment&DP\\\midrule %
		$[0, 10]$&\pgfutilensuremath {0.0634}&\pgfutilensuremath {0.0004}&\pgfutilensuremath {0.0381}&\pgfutilensuremath {0.0003}\\%
		$[0, 100]$&\pgfutilensuremath {0.0652}&\pgfutilensuremath {0.0027}&\pgfutilensuremath {0.0408}&\pgfutilensuremath {0.0023}\\%
		$[0, 1\,000]$&\pgfutilensuremath {0.0621}&\pgfutilensuremath {0.0291}&\pgfutilensuremath {0.0439}&\pgfutilensuremath {0.0255}\\%
		$[0, 10\,000]$&\pgfutilensuremath {0.0757}&\pgfutilensuremath {0.2454}&\pgfutilensuremath {0.0394}&\pgfutilensuremath {0.2572}\\%
		$[0, 100\,000]$&\pgfutilensuremath {0.0763}&\pgfutilensuremath {2.8271}&\pgfutilensuremath {0.0525}&\pgfutilensuremath {3.0296}\\\bottomrule %
		\end {tabular}%
	\end{table}
	
	For $T = 2$, the dynamic program described in Section~\ref{sec:DP} provides an alternative approach to obtain an optimal schedule. To test the computational performance of this dynamic program, we have performed additional experiments using the same set of instances as described above, except that we have re-generated the costs from intervals of increasing sizes (instead of the $[0,10]$ interval as used previously) to examine how this affects performance.   Table~\ref{tab:comp_res_DP} reports on the results of these additional experiments. We focus on the assignment-based formulation and on instances with $m = 10$ because, for this setting, we can solve all instances for both TCTP and TCSP. The table clearly shows that the required CPU time for the dynamic program is roughly proportional to the interval size, which is not surprising given the theoretical pseudo-polynomial running time of $O(n^2c(N))$. The assignment-based formulation, in contrast, seems to be mostly unaffected by the interval size, and it outperforms the dynamic program for intervals of sizes $10\,000$ and $100\,000$.

\subsection{Performance of the local search procedure}\label{sec:perf_ls}
\begin{table}[t]
	\footnotesize
	\centering
	\renewcommand{\arraystretch}{0.6}
	\caption{Computational performance of our local search for TCTP and TCSP.\@ The average and maximal CPU time (both in seconds) are computed over all instances for which an optimal solution is known. The column \%opt gives the percentage of instances for which the local search found the optimum, and the average and maximal percentage optimality gap is computed over instances for which the local search did not find the optimal solution.} \label{tab:comp_ls}
	\begin {tabular}{crrrrr}%
	\toprule \multicolumn {1}{c}{Problem}&\multicolumn {1}{c}{avg CPU}&\multicolumn {1}{c}{max CPU}&\multicolumn {1}{c}{\%opt}&\multicolumn {1}{c}{avg \%gap}&\multicolumn {1}{c}{max \%gap}\\\midrule %
	TCTP&\pgfutilensuremath {0.026}&\pgfutilensuremath {0.064}&\pgfutilensuremath {100.00}&\pgfutilensuremath {0.000}&\pgfutilensuremath {0.000}\\%
	TCSP&\pgfutilensuremath {0.039}&\pgfutilensuremath {0.070}&\pgfutilensuremath {96.36}&\pgfutilensuremath {0.050}&\pgfutilensuremath {0.131}\\\bottomrule %
	\end {tabular}%
\end{table}

\colorlet{graphred}{red!70!black}
\colorlet{graphorange}{orange!70!black}
\colorlet{graphyellow}{yellow!40!black}
\colorlet{graphblue}{cyan!70!black}

\colorlet{Gcol}{graphred}
\colorlet{GLScol}{graphorange}

\tikzset{
	T10/.style={
		graphyellow,
	},
	T15/.style={
		graphred,
	},
	T20/.style={
		graphblue,
	},
}

\pgfplotsset{
	boundsgraph/.style={%
		width=6cm,
		height= 6cm,
		clip=false,
		axis line style={black},
		tick label style={black},
		%
		axis y line= {left},
		y axis line style = {xshift = -1ex},
		y tick label style={xshift = -1ex,
		 /pgf/number format/.cd,
		fixed,
		fixed zerofill,
		precision=2,
		/tikz/.cd
		},
		y tick style={xshift = -1ex,} ,  
		ylabel style={
			rotate=-90, 
			align=left,
			at={(ticklabel* cs:1)},
			anchor= south west,
			xshift = -5ex,
			black,
		},
		%
		axis x line=left,
		x axis line style= {yshift = -1ex},
		x tick label style={yshift = -1ex},
		x tick style={yshift = -1ex} ,  
		xlabel style={
			at={(ticklabel* cs:1)},
			anchor= north west,
			yshift = -1ex,
			align = left,
			black,
		},
	},
}
\pgfplotstableread[col sep=comma]{
	m,T5,T10,T15
	2,0.4599,0.2764,0.1591
	6,0.3838,0.4497,0.4391
	10,0.2014,0.3041,0.3554
	14,0.2351,0.5700,1.0100
	18,0.2043,1.1553,1.2046
}\lsTCSPTable
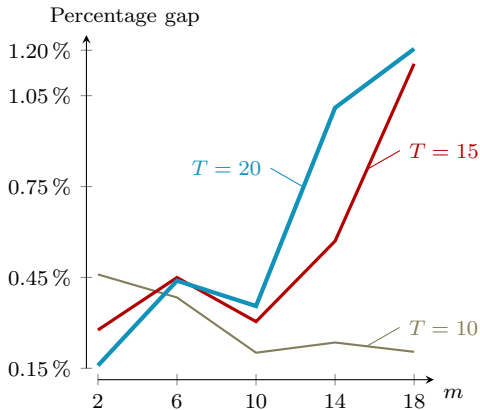
\begin{figure}
	\centering
	\caption{Average percentage gap between the upper bound provided by the local search procedure and the lower bound provided by the LP relaxation of our formulations, for large TCSP instances.} \label{fig:qual_bounds}
	\begin{tikzpicture}[baseline, font = \footnotesize]
		\begin{axis}[
		boundsgraph,
		xmin = 2,
		xmax = 19,
		xtick = {2, 6, 10, 14, 18},
		xlabel={$m$},
		ymin = 0.15,
		ymax = 1.25,
		ytick = {0.15, 0.45, 0.75, 1.05, 1.20},
		ylabel={Percentage gap},
		yticklabel={\pgfmathprintnumber{\tick}\,\%},
		]
		\addplot[thick, T10] table[x index = 0, y index=1]\lsTCSPTable coordinate [pos=0.85] (T10);
		\addplot[very thick, T15] table[x index = 0, y index=2]\lsTCSPTable coordinate [pos=0.85] (T15);
		\addplot[ultra thick, T20] table[x index = 0, y index=3]\lsTCSPTable coordinate [pos=0.65] (T20);				
		\begin{scope}
		\path (T10) ++(30:0.5cm) coordinate (T10a);
		\draw[T10] (T10) -- (T10a);
		\node[anchor = west, T10] at (T10a) {$T=10$};
		\path (T15) ++(30:0.5cm) coordinate (T15a);
		\draw[T15] (T15) -- (T15a);
		\node[anchor = west, T15] at (T15a) {$T=15$};
		\path (T20) ++(150:0.5cm) coordinate (T20a);
		\draw[T20] (T20) -- (T20a);
		\node[anchor = east, T20] at (T20a) {$T=20$};
		\end{scope}
		\end{axis}
	\end{tikzpicture}
\end{figure}

Table~\ref{tab:comp_ls} indicates that the local search described in Section~\ref{sec:heur}, with the three initializations as discussed in Section~\ref{sec:implementation}, finds near-optimal solutions in very limited computation times. In fact, the local search is able to find the optimal solution for all TCTP instances solved to optimality by the MIP formulations. For TCSP, it finds the optimum for~$96.47\%$ of the solved instances and, for those instances where it does not find the optimum, the maximal optimality gap is only~$0.131\%$.

To evaluate the performance of our local search for large TCSP instances, Figure~\ref{fig:qual_bounds} shows the average percentage gap between the upper bound provided by the local search procedure and the lower bound provided by the LP relaxation of our formulations. Since the LP relaxation for the partial-order-based formulation could not always be solved to optimality within the 30-minute time limit, we added Constraints~\eqref{mip:po_transd} iteratively in a cutting-plane fashion; the considered lower bound is then the best  bound obtained within the time limit.  The gap displayed in Figure~\ref{fig:qual_bounds} thus overestimates the true locality gap as it is based on a lower bound for the optimal search cost. As even this overestimate never exceeds 1.2\%, one can infer that our local search provides  high-quality solutions for large TCSP instances as well. Here, the relatively larger gaps for~$T = 15$ and~$T = 20$ when $m \geq 14$ can be partially explained by the fact that we could not solve the LP relaxation to optimality for such large instances. We do not report a similar analysis for TCTP, as our TCTP formulations are significantly weaker than the TCSP formulations, and the lower bound provided by their LP relaxations is not sufficiently tight to assess the performance of the local search.

\section{Conclusion}\label{sec:conc}
We have introduced the time-critical testing problem and the time-critical search problem. These natural generalizations of the classical sequential testing problem and search problem are applicable in a wide range of time-critical operations, such as machine maintenance, medical diagnosis, and new product development. We have shown that both problems are NP-hard, we have developed a pseudo-polynomial dynamic program as well as a FPTAS for the special case of two time slots, and we have described two mixed integer programming formulations as well as a local search heuristic for the general case. With minor modifications, these formulations can also solve closely related testing problems for which no exact and compact mixed integer programming formulation was previously available.  Based on extensive computational experiments, we find that our assignment-based formulation performs better than our partial-order-based formulation for the testing variant, while this is the other way round for the search variant. Despite its lack of a theoretical worst-case guarantee, we also find that, empirically, our local search procedure performs very well in finding near-optimal solutions in limited computation times.

Although we have established that both problems are strongly NP-hard for every fixed number of machines~$m \geq 3$, their complexity status is still open for~$m = 2$. An interesting open question is thus whether there exists a polynomial-time algorithm for the special case of two machines, and whether the search and testing variant still have the same complexity status for this special case. A second open question is whether TCTP is weakly or strongly NP-hard in the special case of a fixed deadline~$T \geq 3$. Corollary~\ref{col:TCTP_NPhard} combined with the pseudo-polynomial dynamic program of Section~\ref{sec:DP} imply that the problem is weakly NP-hard for $T = 2$, but it does not seem straightforward to generalize the dynamic program to arbitrary but fixed~$T$.

Another promising direction for further research is to develop approximation algorithms for both problem variants. For the testing variant, we have illustrated that a greedy approach that iteratively tests a subset with minimal cost-to-failure-probability as well as a pairwise-interchange-based local search heuristic fail to provide a constant-factor approximation guarantee. The approaches of \cite{daldal2016approximation} and \cite{segev2018polynomial} for the batch testing problem also do not seem to be directly modifiable towards our setting. For the search variant, in turn, an analogous cost-to-probability-based greedy approach also does not provide a constant-factor approximation guarantee, but it is still an open question whether our local search heuristic provides a constant-factor approximation for this variant.


\clearpage
\begin{appendices}

\section{A FPTAS for TCTP with two time slots}\label{app:fptas}
\begin{proof}[\unskip\nopunct] In this appendix, we prove Theorem~\ref{th:fptas} by showing how the pseudo-polynomial dynamic program of Section~\ref{sec:DP} can be converted into a FPTAS by appropriately rounding the testing costs. Without loss of generality, assume that the tests are indexed in non-increasing value of their costs, with the largest probability as tie-breaker. That is, $c_1 \geq \ldots \geq c_n$, and~$p_i \geq p_{i + 1}$ for every~$i\in N$ with~$c_i = c_{i + 1}$. If~$c(S^\star) = 0$, then an optimal schedule is given by~$\sigma = (S, N \setminus S)$ with~$S = \{j, \ldots, n\}$ such that~$j$ is the smallest integer for which $c_{j} = 0$ and~$n - j + 1 \leq m$. Since we can therefore check in polynomial time whether~$c(S^\star) = 0$, we henceforth assume that $c(S^\star) > 0$.

Consider an arbitrary~$\varepsilon > 0$ and call~$n_+$ the largest index for which~$c_{n_+} > 0$. For each~$i = 1,\ldots,n_+$, we consider a modified instance in which the costs are rounded down to the nearest multiple of~$\mu^{(i)} \coloneqq \varepsilon c_i/ n$. That is, we replace the cost of each test~$j \in N$ by a new cost~$c^{(i)}_j \coloneqq \lfloor c_j / \mu^{(i)} \rfloor$ and perform our dynamic program for stages $j \in \{i, \ldots, n\}$ with the modified costs to compute~$v_i(b,m)$ for each~$b = 0,\ldots, \sum_{j = i}^n c^{(i)}_j$. If~$v_i(b,m) = +\infty$ for all values of~$b$, then we define~$z^{(i)}_\varepsilon \coloneqq +\infty$.  Otherwise, we consider an optimal testing set~$S^{(i)}_\varepsilon \subseteq \{i, \ldots, n\}$ and the schedule $\sigma^{(i)}_\varepsilon = (S^{(i)}_\varepsilon, N \setminus S^{(i)}_\varepsilon)$ for the modified instance, and we denote the expected testing cost in the original instance as
\begin{equation*}
z^{(i)}_\varepsilon \coloneqq z(\sigma^{(i)}_\varepsilon) = \sum_{j \in S^{(i)}_\varepsilon}c_j + \sum_{j \in N \setminus  S^{(i)}_\varepsilon}c_j \prod_{j \in S^{(i)}_\varepsilon} p_j.
\end{equation*} 
Observe that~$z^{(i)}_\varepsilon$ is finite for at least one~$i \in \{1,\ldots, n\}$ since~$v_1(b,m) \leq 1$ for~$b = \sum_{j=1}^m c_j^{(1)}$.  The algorithm completes by returning a testing set~$S_\varepsilon = S^{(k)}_\varepsilon$ and schedule~$\sigma_\varepsilon = (S_\varepsilon, N \setminus S_\varepsilon)$, where~$k$ is such that~$\sigma^{(k)}_\varepsilon$ achieves the minimum cost $z^{(k)}_\varepsilon = \min_{i = 1,\ldots, n_+} z^{(i)}_\varepsilon$ for the original instance.

Since~$c_1 \leq \ldots \leq c_n$, we have for each~$i = 1,\ldots,n_+$ that
\begin{equation*}
\sum_{j = i}^n c^{(i)}_j = \sum_{j = i}^n \left\lfloor \frac{c_j}{\mu^{(i)}} \right\rfloor = \sum_{j = i}^n \left\lfloor \frac{n c_j}{\varepsilon c_i} \right\rfloor \leq \sum_{j = i}^n \frac{n c_j}{\varepsilon c_i} \leq \sum_{j = i}^n \frac{n}{\varepsilon}  \leq \frac{n^2}{\varepsilon},
\end{equation*}
and our dynamic program yields~$S^{(i)}_\varepsilon$ for a given~$i$ in time~$O(n^4 / \varepsilon)$. Since we need to evaluate~$z^{(i)}_\varepsilon$ for at most~$n$ different values of~$i$, we thus obtain~$S_\varepsilon$ and~$\sigma_\varepsilon$ in time $O(n^5 / \varepsilon)$.

It remains to be shown that $z(\sigma_\varepsilon) \leq (1 + \varepsilon) z(\sigma^\star)$. Let~$i$ be the largest index for which~$S^\star \subseteq \{i, i+1, \ldots, n\}$. Observe that~$z^{(i)}_\varepsilon$ is finite since~$S^\star$ is feasible, and that~$z(\sigma^\star) \geq c(S^\star) \geq c_i > 0$ because $i \in S^\star$ and $c(S^\star) > 0$. It also follows from our definition of~$S^{(i)}_\varepsilon$ as an optimal testing set in the modified instance that
\begin{equation*}
\sum_{j \in S^{(i)}_\varepsilon} c^{(i)}_j +  \sum_{j \in N \setminus S^{(i)}_\varepsilon} c^{(i)}_j \prod_{j \in S^{(i)}_\varepsilon} p_j \leq \sum_{j \in S^\star} c^{(i)}_j + \sum_{j \in N \setminus S^\star} c^{(i)}_j \prod_{j \in S^\star} p_j.
\end{equation*}
Finally, observe that, by construction, $c^{(i)}_j = \lfloor c_j / \mu^{(i)}\rfloor \leq c_j / \mu^{(i)} \leq \lfloor c_j / \mu^{(i)}\rfloor + 1$ for every~$j \in N$. From the observations above and the definition of~$\mu^{(i)} = \varepsilon c_i /n$, we obtain that
\begin{align*}
z(\sigma^{(i)}_\varepsilon) &= \sum_{j \in S^{(i)}_\varepsilon}c_j + \sum_{j \in N \setminus  S^{(i)}_\varepsilon}c_j \prod_{j \in S^{(i)}_\varepsilon} p_j \\
&\leq \sum_{j \in S^{(i)}_\varepsilon} \mu^{(i)}\left(\left\lfloor\frac{c_j}{\mu^{(i)}}\right\rfloor + 1\right) + \sum_{j \in N \setminus S^{(i)}_\varepsilon} \mu^{(i)} \left(\left\lfloor\frac{c_j}{\mu^{(i)}}\right\rfloor + 1\right) \prod_{j \in S^{(i)}_\varepsilon} p_j \\
&\leq \mu^{(i)} \left(\sum_{j \in S^{(i)}_\varepsilon} \left\lfloor\frac{c_j}{\mu^{(i)}}\right\rfloor + \sum_{j \in N \setminus S^{(i)}_\varepsilon} \left\lfloor\frac{c_j}{\mu^{(i)}}\right\rfloor \prod_{j \in S^{(i)}_\varepsilon} p_j \right) + n \mu^{(i)} \\
&\leq \mu^{(i)} \left( \sum_{j \in S^\star} \left\lfloor\frac{c_j}{\mu^{(i)}}\right\rfloor + \sum_{j \in N \setminus S^\star} \left\lfloor\frac{c_j}{\mu^{(i)}}\right\rfloor \prod_{j \in S^\star} p_j \right) + n \mu^{(i)} \\
&\leq \left(\sum_{j \in S^\star} c_j + \sum_{j \in N \setminus S^\star} c_j \prod_{j \in S^\star} p_j \right) + \varepsilon  c_i \\
&\leq z(\sigma^\star) + \varepsilon z(\sigma^\star) = (1+ \varepsilon)z(\sigma^\star).
\end{align*}
Since~$z(\sigma_\varepsilon) = \min_{j = 1,\ldots, n_+} z(\sigma^{(j)}_\varepsilon)$, this shows that~$z(\sigma_\varepsilon) \leq z(\sigma^{(i)}_\varepsilon)\leq (1 + \varepsilon) z(\sigma^\star)$.
\end{proof}
\end{appendices}

\end{document}